\documentclass{article}
\usepackage[utf8]{inputenc}

\pdfoutput=1
\usepackage{natbib}
\setcitestyle{authoryear}
\usepackage[english]{babel}
\usepackage[nottoc]{tocbibind}

\usepackage[dvips]{color}
\usepackage{xcolor}
  
\usepackage{amsmath,amsthm,amssymb,graphicx,comment,dsfont,fullpage, mathtools, tabu, multirow, enumitem, url, indentfirst}
\usepackage{wrapfig}
\usepackage{algorithm2e, pdflscape, makecell,sidecap}
\usepackage{authblk}
\RestyleAlgo{ruled}

\usepackage{hyperref}
\hypersetup{linkcolor=blue, citecolor=blue, colorlinks=true}

\usepackage{caption}
\usepackage{float}

\AtBeginDocument{}

\newtheorem{theorem}{Theorem}
\newtheorem{lemma}[theorem]{Lemma}
\newtheorem{proposition}[theorem]{Proposition}

\title{CRP-Tree: A phylogenetic association test for binary traits}
\author[1]{Julie Zhang}
\author[2]{Gabriel A. Preising}
\author[2]{Molly Schumer}
\author[1,3]{Julia A. Palacios}
\date{February 17,  2023}

\affil[1]{Department of Statistics, Stanford University}
\affil[2]{Department of Biology, Stanford University}
\affil[3]{Department of Biomedical Data Science, Stanford University}

\begin{document}

\maketitle
\begin{abstract}
An important problem in evolutionary genomics is to investigate whether a certain trait measured on each sample is associated with the sample phylogenetic tree. The phylogenetic tree represents the shared evolutionary history of the samples and it is usually estimated from molecular sequence data at a locus or from other type of genetic data. We propose a model for trait evolution inspired by the Chinese Restaurant Process that includes a parameter that controls the degree of preferential attachment, that is, the tendency of nodes in the tree to subtend from nodes of the same type. This model with no preferential attachment is equivalent to a structured coalescent model with simultaneous migration and coalescence events and serves as a null model. We derive a test for phylogenetic binary trait association with linear computational complexity and empirically demonstrate that it is more powerful than some other methods. We apply our test to study the phylogenetic association of some traits in swordtail fish, breast cancer, yellow fever virus and influenza A H1N1 virus. R package implementation of our methods is available at \texttt{https://github.com/jyzhang27/CRPTree}. 
\end{abstract}


\newpage 

\tableofcontents
\newpage

\section{Introduction}\label{sec:1}

Understanding the genetic basis of phenotypic traits is a fundamental goal in evolutionary biology and molecular epidemiology of infectious diseases. In particular, an important question is whether a certain observed trait is associated with the phylogenetic tree structure at a certain locus. Traits include geographic location, disease susceptibility, physical characteristics, and behavioral traits. To give a concrete example in infectious diseases, we can consider the phylogenetic tree (Figure~\ref{fig:yfw_post_figure}(b)) reconstructed from yellow fever virus (YFV)  molecular sequences obtained from infected humans (blue tips) and nonhuman primates (red tips) in South America. It is of interest to investigate whether the virus has been spreading only within each population, that is, whether human samples are more closely related to other human samples than to nonhuman primate samples. The answer to this question would provide insight about the recent outbreak of YFV in South America \citep{Faria2018}. This relation is known as \textbf{phylogenetic trait association} or \textbf{phylogenetic signal}. That is, the tendency of related organisms to share some trait characteristics more than organisms drawn at random from the same tree. Phylogenetic trait association is also sometimes assessed prior to doing a comparative analysis between traits or phenotypes. For example, one can use phylogenetically independent contrasts \citep{Felsenstein1985, Garland1992} to compare the association between two phenotypes given the phylogeny.\\

To illustrate the problem in phylogenetic trait-association, consider the following two trees in Figure~\ref{fig:trait_ex}. They have the same tree topology, but the tips have different trait values. In Figure~\ref{fig:trait_ex}(a), we see that all the nodes of the same type are in one subtree, and so clearly the trait is associated with the tree structure. In Figure~\ref{fig:trait_ex}(b), it is not so obvious whether there is a relation because the types are scattered throughout the tree topology. Later, we will illustrate that there is no phlyogenetic trait association in tree in Figure~\ref{fig:trait_ex}(b) according to our test. \\

\begin{figure}[h]
    \centering
    \includegraphics[width=0.95\textwidth]{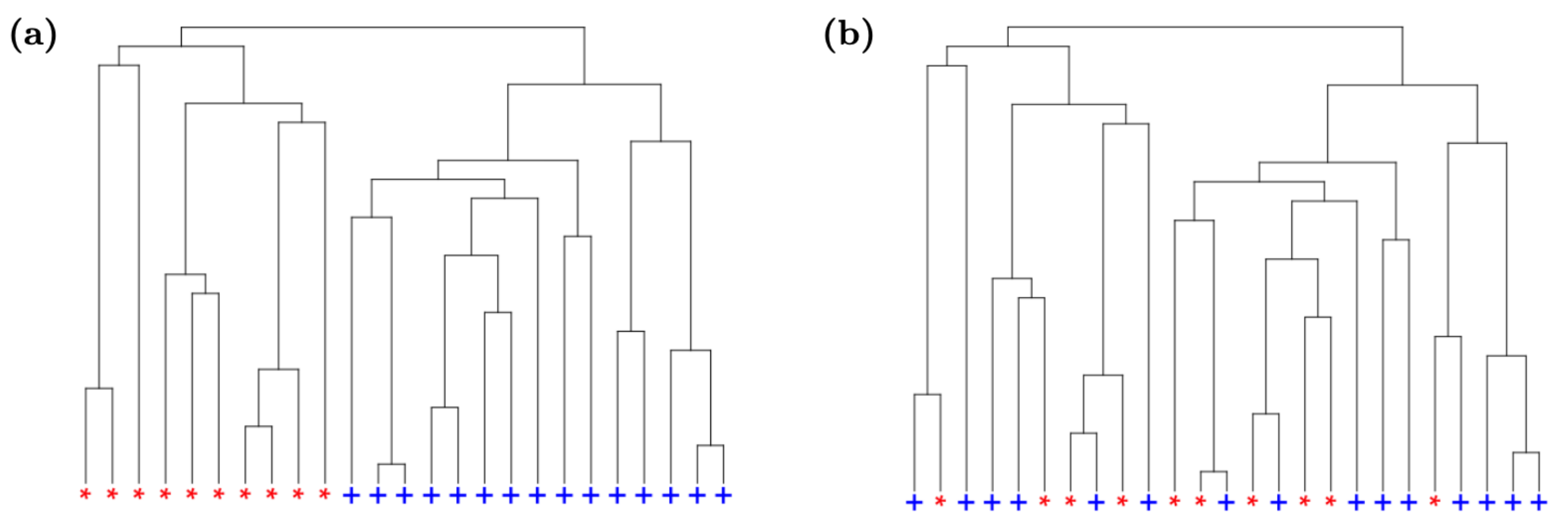}
    \caption{\textbf{Examples of phylogenetic trait associations.} Two trees with 25 tips, 10 of type 1 and 15 of type 2. \textbf{(a)} Clear association between the trait values (colors at the tips) and the tree topology. \textbf{(b)} We cannot tell visually if there is any association. Our proposed test does not detect association.}
    \label{fig:trait_ex}
\end{figure}

Many methods for discrete phylogenetic trait association have been developed in the fields of phylogenetics and molecular epidemiology. A large class of these methods consists in calculating a single summary statistic that conveys information about the trait-phylogeny association. Significance is then assessed through a permutation p-value obtained by permuting the leaf labels. The parsimony score (PS) is one of the most commonly used statistic. It counts the minimum number of trait state changes needed in the phylogeny in order to reconstruct character states at ancestral nodes \citep{Fitch1971,hartigan1973minimum,slatkin1989}. It is closely related to the Maximum Parsimony tree, which is the tree that minimizes the PS statistic. \cite{Wang2001} propose the association index (AI) statistic that measures the imbalance of the internal phylogeny nodes, and \cite{Borges2019} propose a measure based on Shannon entropy per node. Phylogenetic diversity \citep{Faith1992}, nearest taxa index and nearest relatedness index \citep{Webb2000, Webb2002}, and UniFrac \citep{Lozupone2005} are statistics that utilize branch length information and tree topology to capture trait-phylogeny association. However, a major drawback of the aforementioned class of methods is that they ignore phylogenetic uncertainty. Phylogenetic trees are usually estimated from molecular data with large uncertainty. In order to solve this problem, \cite{Parker2008} propose a method called BaTS (Bayesian Tip-association Significant testing), that incorporates phylogenetic uncertainty in a Bayesian Markov chain Monte Carlo framework. The authors use the posterior median of the association statistic to assess significance by comparing its value to the approximated null distribution of the median association statistic obtained by random permutation of the label set. \\

Another class of methods is based on change point detection along the phylogeny \citep{Ansari2016,Behr2020}. In \cite{Behr2020}, binary trait values at the tips of node $i$ are assumed to be independent Bernoulli random variables with success probability $p_{i}$. The goal is then to detect the internal nodes at which the success probability changes via likelihood ratio statistics. The authors suggest their method can be extended to categorical and continuous traits. However, these methods also assume that the phylogeny is known without uncertainty. Phylogenetic trait-association methods for continuous traits usually model trait states via tree-valued Gaussian processes. \cite{Munkemuller2012} provides an extensive review of these methods and are not considered here. \\

In this article, we develop a phylogenetic trait association test for binary traits. Our method is applicable to fixed phylogenies and to latent phylogenies within a Bayesian Markov chain Monte Carlo framework. We propose a model for trait evolution inspired by the Chinese Restaurant Process \citep{Aldous1985} that depends on a single parameter $\alpha$. The model provides a tree-generating process in which the likelihood of lineages to descend from lineages of the same type is controlled by $\alpha$. In this model, the number of same-type attachments is a sufficient statistic for $\alpha$, therefore, our test statistic uses this information. Having a general model of trait evolution allows us to empirically test the power of the test under a large family of alternatives.\\

We structure our article as follows. First, in Section~\ref{sec:2} we will introduce some terminology and definitions of the different tree topologies to be analyzed later, together with some novel and known enumerative results. We propose a coalescent model on partially labeled phylogenetic trees in Section~\ref{sec:3}. This model serves as the null model for our testing problem. In Section~\ref{sec:4}, we introduce the CRP-Tree model inspired by the Chinese Restaurant Process and state several results. We provide our test statistic and discuss how to assess significance in both a fixed tree case and a Bayesian framework in Section~\ref{sec:5}. In Sections~\ref{sec:6} and \ref{sec:7}, we analyze the performance of our test in simulated and real data applications. Finally, in Section~\ref{sec:8} we summarize our contributions and discuss future directions.

\section{Preliminaries}\label{sec:2}

We first describe the four types of phylogenetics trees considered in this manuscript and provide some enumerative properties that will be used later. All trees are rooted and binary, and we only consider their tree topology ignoring branch length information. 

\begin{enumerate}
    \item A \textbf{ranked tree shape} $T_N$ is a rooted, binary tree shape with $N$ unlabeled tips and a total ordering of the internal nodes. The number of such trees is given by the Euler zig-zag numbers $e(N)$, defined via the recurrence relation \citep{Murtagh1984}. 
    \begin{equation}\label{eq:zig-zag}
        e(N) = \frac{1}{2} \sum_{k=0}^{N-2} \binom{N-2}{k} e(k+1) e(n-k-1).
    \end{equation}
    The base cases are $e(1)=e(2)= 1$. To see this, let $T^L, T^R$ denote the left and right subtrees of $T_N$, and suppose $T^L$ has $k$ internal nodes and $k+1$ tips, while $T^R$ has $n-k-2$ internal nodes and $n-k-1$ tips where $0\leq k \leq n-2$. Then there are $\binom{n-2}{k}$ ways to arrange the internal nodes of both subtrees in order. In addition, there are a total of $e(k+1)$ possibilities for $T^L$ and $e(n-k-1)$ possibilities for $T^R$. Accounting for the symmetry of $T^L$ and $T^R$ gives the final formula Equation (\ref{eq:zig-zag}). The first elements of the sequence are $1, 1, 1, 2, 5, 16, 61$. 

    \item  A \textbf{ranked planar tree shape} $\tilde{T}_N$ is a ranked tree shape with $N$ unlabeled tips where the left and right child nodes of an internal node are distinguished. The number of such trees is the number of permutations of $\{1,2,...,N-1\}$, that is $(N-1)!$ \citep{Cleary2015}. To see this, note that the order of appearance of internal node labels from left to right defines a ranked planar tree shape. For example, if $N=4$, there are $(4-1)!=6$ ranked planar tree shapes given by the 6 orderings of $\{1,2,3\}$. Figure~\ref{fig:ranked_planar_4} shows the 6 trees and their internal node orderings. 

    \begin{figure}[h]
        \centering
        \includegraphics[width=0.9\textwidth]{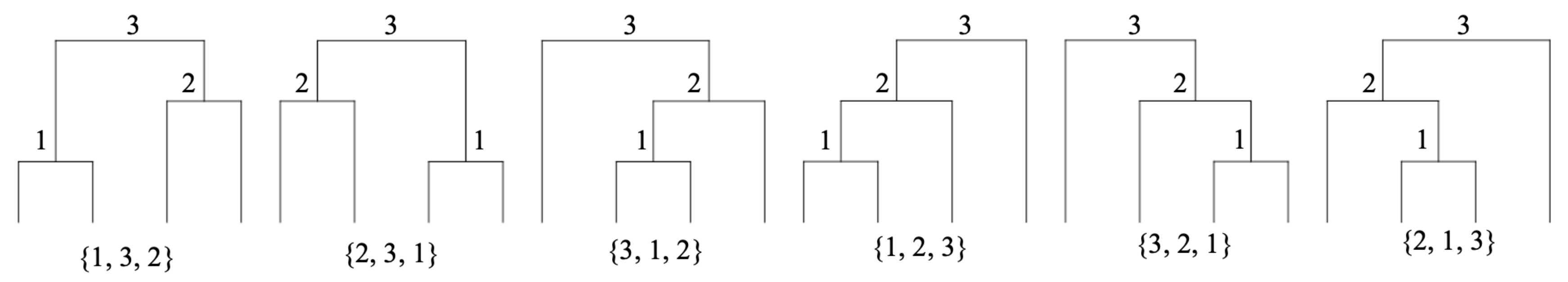}
        \caption{\textbf{Example of all 6 ranked planar tree shapes with 4 tips.} Each tree is identified by the order of its internal nodes from left to right.}
        \label{fig:ranked_planar_4}
    \end{figure}
    
    \item  A \textbf{ranked partially labeled tree} $T^{\ell}_{N,B}$ with $N$ tips is a ranked tree shape whose tip labels form a multiset. In particular, we will focus on multisets with two unique elements, which we denote by $\{$Blue, Red$\}$ for the rest of this manuscript. We can count the number of such trees $R(N,B)$ by a recursion similar to the number of ranked tree shapes. 
    \begin{equation}\label{eq:ranked-partially-labeled}
        R(N,B) = \frac{1}{2} \sum_{n=1}^{N-1} \sum_{b=0}^{\min(B,n)} R(n,b) R(N-n, B-b) \binom{N-2}{n-1}.
    \end{equation}
    The base cases are $R(1,1) = R(1,0)= 1, R(2,2) = R(2,0) = R(2,1) =1$. The derivation is parallel to that of the Euler zig-zig numbers. Let the left subtree of $T^{\ell}_{N,B}$ have $n$ tips. Then there are $\binom{N-2}{n-1}$ ways to arrange the internal nodes of both subtrees in order ($n-1$ in the left, and $N-n-1$ in the right). The term $R(n,b) R(N-n, B-b)$ counts the number of
    trees in which the left subtree of $T^{\ell}_{N,B}$ has $b$ blue tips and $n$ total tips, and the right subtree of $T^{\ell}_{N,B}$ has $B-b$ blue tips and $N-n$ total tips. Summing over all possible $n$ and accounting for the symmetry of the left and right subtrees gives the final Equation (\ref{eq:ranked-partially-labeled}). Note also that $R(N,0) = R(N,N) = e(N)$. Starting from $N=3$, the first few values of $R(N,\lfloor N/2 \rfloor)$ are $2, 7, 27, 152, 935$.  

    \item A \textbf{ranked planar partially labeled tree} $\tilde{T}^{\ell}_{N,B}$ is a ranked planar tree shape with $N$ leaves, and $B$ tips labeled blue, $N-B$ tips labeled red. There are $(N-1)! \binom{N}{B}$ such trees because there are $\binom{N}{B}$ possible labelings on every ranked planar tree. 
\end{enumerate}

We use the superscript $\ell$ in $T_{N,B}^\ell$ to indicate the partial labeling, and the tile in $\tilde{T}^{\ell}$ to indicate the tree is planar. In addition, for any tree $T$ (regardless of resolution), we will use $C(T)$ to denote the number of cherries of $T$, that is, the number of subtrees with exactly two tips. We will use $C_S(T^\ell)$ to denote the number of cherries of $T^\ell$ with the same label (regardless of resolution). Then for a given ranked tree shape $T_N$, there are $2^{N-1-C(T_N)}$ ranked planar trees $\tilde{T}_N$. This is because there are $N-1-C(T_N)$ nodes with distinct left and right subtrees that can be swapped to generate new planar trees. Similarly, for a given ranked partially labeled tree $T^{\ell}_{N,B}$, there are $2^{N-1-C_S(T_N)}$ ranked planar partially labeled trees $\tilde{T}^{\ell}_{N,B}$. Therefore, for a given ranked tree shape $T_N$, there are $2^{N-1-C(T_N)} \times \binom{N}{B}$ ranked, partially labeled, planar trees. Figure~\ref{fig:tree_resolutions3} displays the trees with $N=3, B=1$ in the four resolutions. \\

\begin{figure}[h]
    \centering
    \includegraphics[width=0.75\textwidth]{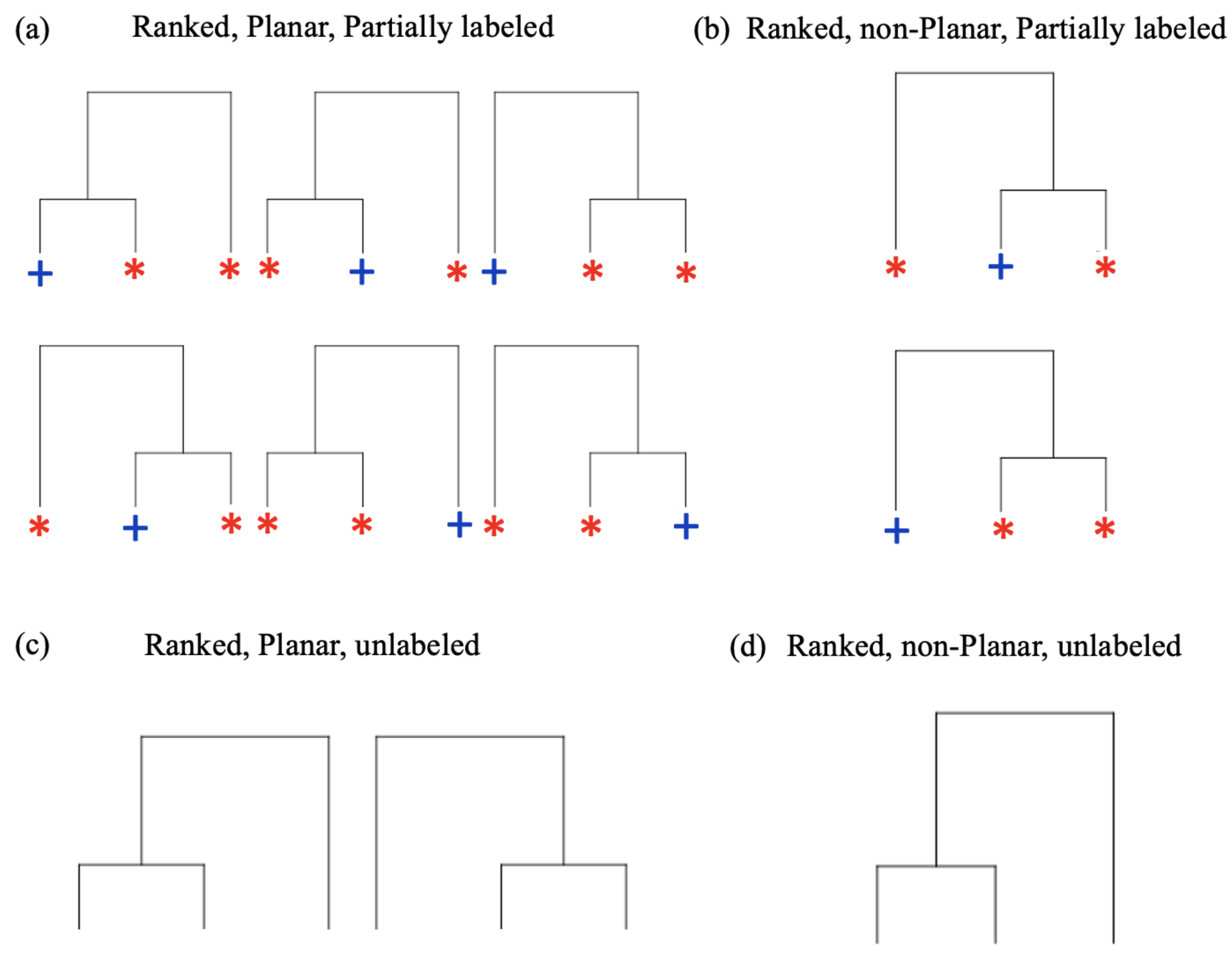}
    \caption{\textbf{Four different tree resolutions for trees with three tips}. The colors (and shapes for clarity) on the tips represent the tip labels, with $B=1$. We see that there are 6 ranked planar partially labeled trees $\tilde{T}^{\ell}_{N,B}$, 2 ranked planar tree shapes $\tilde{T}_N$, 2 ranked partially labeled tree shapes $T^{\ell}_{N_B}$, and 1 ranked tree shape $T_N$.}
    \label{fig:tree_resolutions3}
\end{figure}

A ranked partially labeled tree with $N$ tips and $B$ blue, that has a monophyletic clade with respect to at least one color (blue or red) will be called  a \textbf{perfect tree shape} and denoted by $T^{\ell, P}_{N,B}$. We say a tree is \textbf{exactly perfect} $T^{\ell,EP}_{N,B}$ if there is monophyly with respect to both colors. That is, all red tips and all blue tips are contained respectively in the two subtrees descending from the root. For example, Figure~\ref{fig:trait_ex}(a) is an exactly perfect tree shape, while a caterpillar tree with the tips of the cherry being blue and all other tips being red is a perfect tree shape. We extend these definitions to $\tilde{T}^{\ell, P}_{N,B}, \tilde{T}^{\ell, EP}_{N,B}$ if in addition, the tree is planar. Perfect trees are in a sense the most extreme partial labeling possible that separates the two types on a ranked tree shape. 

\begin{proposition}
The number of exactly perfect tree shapes with $N$ tips and $B$ blue is \[ EP(N,B) = e(B)\times e(N-B)\times \binom{N-2}{B-1}. \] 
The number of perfect tree shapes given $N$ and $B$ is defined recursively as 
\begin{align*}
     P(N,B) &= \left ( \sum_{i=0}^{N-B-1} e(N-B-i) P(B+i, B) \binom{N-2}{B+i-1} \right ) + \\
     &\qquad \left ( \sum_{i=0}^{B-1} e(B-i) P(N-B+i, N-B) \binom{N-2}{N-B+i-1} \right ) -  e(B)\times e(N-B)\times \binom{N-2}{B-1}.
\end{align*}
The base cases are $P(n,n)= e(n), P(n+1, n) = e(n)$. 
\end{proposition}
\begin{proof}
To prove the first statement note that there are $\binom{N-2}{B-1}$ ways to interleave the internal nodes of the two monophyletic subtrees together, each with $e(B)$ and $e(N-B)$ ranked tree shapes. \\

\noindent To find the number of perfect tree shapes, first suppose there is a blue monophyletic clade. Let the subtree containing the blue monophyletic clade have a total of $B+i$ tips, with $i=0,...,N-B-1$. There are a total of $P(B+i, B)$ such subtrees. The other subtree has $N-B-i$ red tips with $e(N-B-i)$ possible ranked tree shapes. Since there are $\binom{N-2}{B+i-1}$ ways to interleave both subtrees, we then get the first term in the summation. The second summand is obtained equivalently considering the red monphyletic clade. To get the final answer, we must subtract the number of exactly perfect trees to correct for double counting. 
\end{proof}

\noindent\textbf{Remark:} For any $T^{\ell}_{N,B}$, we would ideally like to know how ``extreme'' an observed partial labeling on a ranked tree shape $T_N$ is with respect to the uniform distribution on label assignments. Though we can define a ``most extreme'' partial labeling, there is not a clear way to define an ordering among possible partial labelings, and therein lies the difficulty of phylogenetic trait association. 

\section{A null coalescent model} \label{sec:3}
We describe a coalescent model on ranked partially labeled tree shapes with $N$ tips, such that $B$ leaves are blue and $R=N-B$ leaves are red. The model can be described as a bottom-up Markov chain in which every pair of lineages have equal probability of merging. In the tree, every internal lineage is labeled according to the order it is created. The state space can be described as $(r_t,b_t,S_t)$ that records the number of red lineages $r_t$, the number of blue lineages $b_t$, and $S_t$ denotes the set of internal lineages. The full realization $\{(r_t, b_t, S_t)\}_{t=0}^{N-1}$ uniquely encodes a ranked partially labeled tree shape. The initial state at the bottom of the tree is $(R, B,\emptyset)$ and the absorbing state at the root is $(0,0, \{N-1\})$. \\

The initial state (at the tips) has no internal lineages, only blue and red nodes. It then transitions as follows
\begin{equation} \label{eq:trans-prob-1}
(R,B, \emptyset) \rightarrow 
\begin{cases}
(R-2,B, \{ 1\}) & \text{w.p.}\;\; \frac{\binom{R}{2}}{\binom{N}{2}} \\
(R, B-2, \{1\}) & \text{w.p.} \;\;\frac{\binom{B}{2}}{\binom{N}{2}} \\
(R-1, B-1, \{1\}) & \text{w.p.}\;\; \frac{RB}{\binom{N}{2}} 
\end{cases}
\end{equation}

After $t$ steps, the state $(r_t, b_t, S_t)$ indicates the tree has $r_t+b_t+|S_t|= N-t$ extant lineages, of which $r_t$ lineages subtend red leaves, $b_t$ lineages subtend blues leaves, and $|S_t|$ subtend internal nodes. Let $k=|S_t|$ be the number of current lineages subtending internal nodes and $s_i, s_j \leq t$ denote internal nodes that are to be removed (because they will be merged). Then the $(t+1)$th transition for $t>1$ has the following transition probabilities. 

\begin{equation} \label{eq:trans-prob-2}
(r_t, b_t, S_t ) \rightarrow 
\begin{cases}
(r_t-2, b_t, S_t \cup \{t+1 \} ) & \text{w.p.} \;\; \frac{\binom{r_t}{2}}{\binom{r_t+b_t+k}{2}} \\
(r_t, b_t-2, S_t \cup \{t+1 \}) & \text{w.p.} \;\;\frac{\binom{b_t}{2}}{\binom{r_t+b_t+k}{2}} \\
(r_t-1, b_t-1, S_t \cup \{t+1\}) & \text{w.p.} \;\;\frac{r_t b_t}{\binom{r_t+b_t+k}{2}} \\
(r_t, b_t-1, S_t \cup \{t+1 \}\backslash \{s_i\} ) & \text{w.p.}\;\; \frac{b_t}{\binom{r_t+b_t+k}{2}}\\
(r_t-1, b_t, S_t \cup \{t+1 \}\backslash \{s_i\} ) & \text{w.p.}\;\; \frac{r_t}{\binom{r_t+b_t+k}{2}}\\
(r_t, b_t, S_t \cup \{t+1 \}\backslash \{s_i, s_j\} ) & \text{w.p.} \;\;\frac{1}{\binom{r_t+b_t+k}{2}} \\
\end{cases}
\end{equation}

Another way to intuitively understand this jump chain is via an urn process. We start off with $B$ blue balls and $N-B$ red balls in an urn. At the $t$-th iteration, we draw two balls without replacement and add a ball with label $t$ back into the urn. This numbered ball represents  the internal node that was created, while the two balls we removed denote the lineages that were merged. The urn process ends when there is one ball left in the urn, namely ball $N-1$. Figure~\ref{fig:null_model} pictorially demonstrates an example of a full realization of the jump chain starting with 3 blue balls and 2 red balls. 

\begin{figure}[h]
    \centering
    \includegraphics[width=0.75\textwidth]{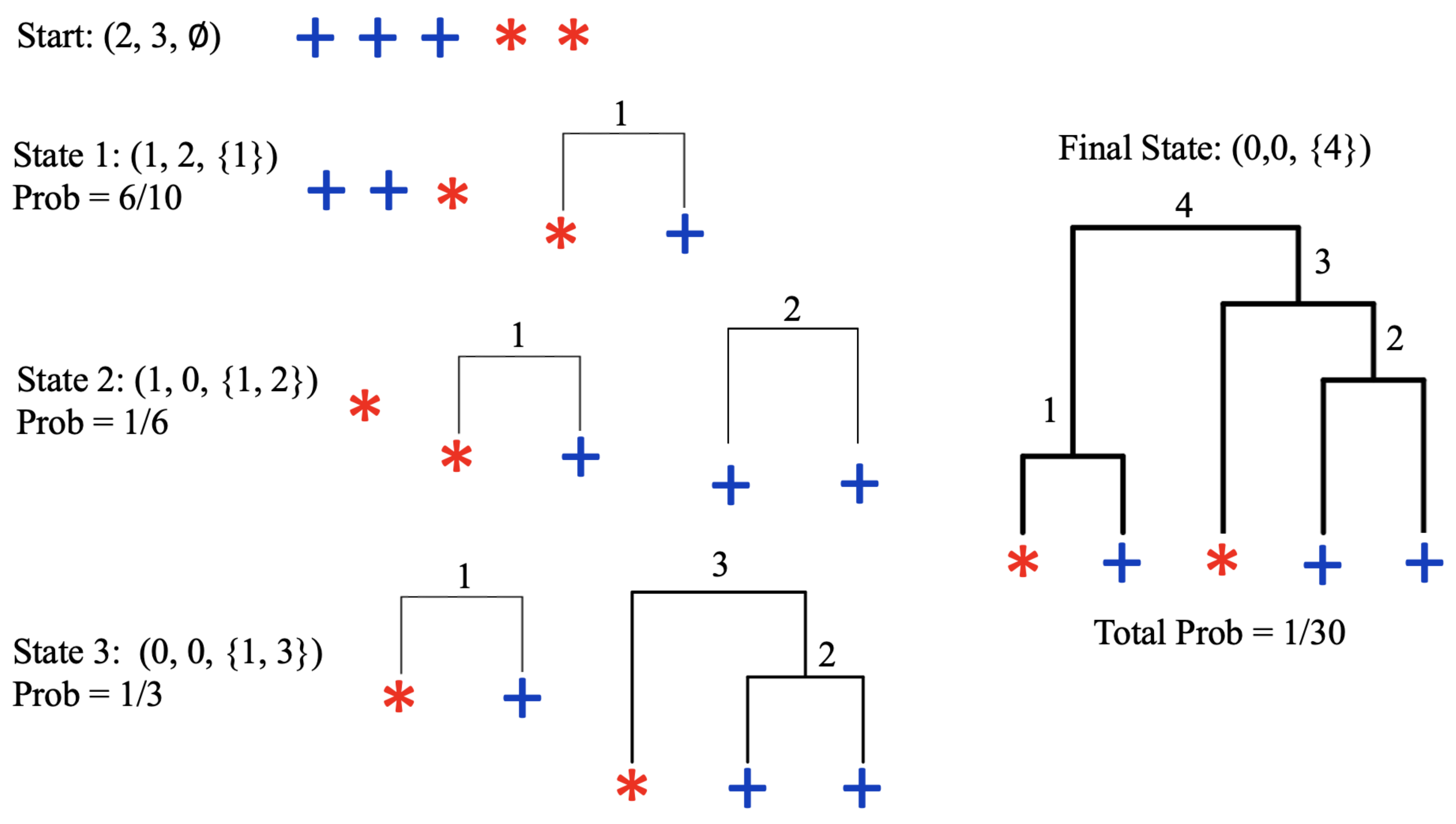}
    \caption{\textbf{One realization of the Markov jump chain starting with 2 red and 3 blue leaves}. At each step, we detail the states, the probability of the transition, and the corresponding step in the tree. The final probability of the tree matches the closed form expression from Theorem~\ref{thm:coal}: $\frac{1}{30}= \frac{2^{5-1-1}}{24\cdot 10}$.}
    \label{fig:null_model}
\end{figure} 

\begin{theorem}\label{thm:coal}
The probability of observing $T^{\ell}_{N,B}$, a ranked partially labeled tree under the null coalescent model is \[ \mathbb{P}(T^{\ell}_{N,B})= \frac{2^{N-C_S(T^{\ell}_{N,B}) -1}}{(N-1)! \binom{N}{B}}, \] where $C_S(T^{\ell}_{N,B})$ is the number of cherries of $T^{\ell}_{N,B}$ with the same label. 
\end{theorem}
\begin{proof}
First, the denominator resulting from the product of all transitions probabilities in Equations (\ref{eq:trans-prob-1}) and (\ref{eq:trans-prob-2}) is \[ \binom{N}{2} \binom{N-1}{2} \cdots \binom{3}{2} \binom{2}{2} = \frac{N! (N-1)!}{2^{N-1}}. \] 
The only transitions that invoke the $\frac{1}{2}$ factor are the coalescent events of two leaves with the same color, hence there is a term $\frac{1}{2^{C_S}}$ in the product of transition probabilities. If we merge two internal nodes, the numerator is multiplied by 1. When a red (blue) leaf is involved in a merger, the numerator of the transition probability is proportional to the number of red (blue) leaves. We then get the factor $B!(N-B)!$. Multiplying all these factors gives us the desired result.
\end{proof}

As a special case, consider all samples being of the same type. In this case, the Markov chain corresponds to the Tajima coalescent \citep{Sainudiin2015,palacios2019bayesian}, and the probability of a specific ranked tree shape is

\[ \mathbb{P}(T_N) = \frac{2^{N-C(T_N)-1}}{(N-1)!}, \] 
where $C(T_N)$ is the number of cherries of $T_{N}$. Indeed, if we sum the probabilities of all the possible partial labelings on a specific ranked tree shape, we get the same probability expression. \\

Note that we can extend this null coalescent model to more than two categories of tip labels. Suppose there are $m$ unique tip labels (i.e. $m$ colors), with $n_1,...,n_m$ number of labels of color $1,...,m$ respectively, which we denote by $[n_i]_m$ for shorthand. Then the analogous extension of Theorem~\ref{thm:coal} is
\begin{equation}\label{thm:coal_multiple}
    \mathbb{P}(T^{\ell}_{N,[n_i]_m})= \frac{2^{N-C_S(T^{\ell}_{N,[n_i]_m}) -1}}{(N-1)! \binom{N}{n_1,...,n_m}}, 
\end{equation}
where $C_S(T^{\ell}_{N,[n_i]_m})$ is the number of cherries of $T^{\ell}_{N,[n_i]_m}$ with the same label and $ \binom{N}{n_1,...,n_m}$ is the multinomial coefficient. In what follows, we will assume only binary labelings. 

\begin{proposition}\label{eq:null_model_cond_prob}
Fix a given ranked tree shape $T_N$, and $B\leq N$. The conditional probability of a specific partial labeling on $T_N$ is \[ \mathbb{P}(T^{\ell}_{N,B} \mid T_N)= \begin{cases}
0 & \text{ if } T^{\ell}_{N,B} \nprec T_{N} \\
\frac{2^{C(T^{\ell}_{N,B})- C_S(T^{\ell}_{N,B})}}{\binom{N}{B}} & \text{ if}  T^{\ell}_{N,B} \prec T_{N},\end{cases}\]
where $T^{\ell}_{N,B} \prec T_{N}$ indicates that $T_{N}$ is obtained from $T^{\ell}_{N,B}$ by removing the leaf labels in $T^{\ell}_{N,B}$.
\end{proposition}
\begin{proof}
We can prove this using our previous result and noting $C(T^{\ell}_{N,B}) = C(T_N)$. 
\begin{align*}
    \mathbb{P}(T^{\ell}_{N,B}\mid T_N) = \frac{P(T^{\ell}_{N,B} \cap T_N)}{P(T_N)} &= \frac{P(T^{\ell}_{N,B})}{P(T_N)} \mathds{1} \{T^{\ell}_{N,B} \prec T_{N}\} \\
    & = \frac{2^{C(T^{\ell}_{N,B})- C_S(T^{\ell}_{N,B})}}{\binom{N}{B}} \mathds{1} \{T^{\ell}_{N,B} \prec T_{N}\} .
\end{align*}

\noindent We directly see that the probability of observing a particular labeling given a ranked tree shape according to the coalescent null model (\ref{eq:trans-prob-2}) is not uniform. Yet, we can generate labeled trees with this probability by random permutation of the leaf labels. 
\end{proof}

The model described is a lumping of the standard coalescent that models completely labeled ranked tree shapes \citep{Kingman1982}. However, this model is finer than the Tajima coalescent that models unlabeled ranked tree shapes \citep{Sainudiin2015}. The model is a modified structured coalescent model without explicit migration, but rather a simultaneous migration and coalescent event can happen in one transition \citep{Notohara1990, Mueller2017}. Currently, we are ignoring branch lengths, so this is a discrete jump process, but one can easily incorporate exponential waiting times as per the standard continuous-time Markov Chain theory. 

\section{The CRP-Tree model}\label{sec:4}
We wish to test whether an observed ranked partially labeled tree is a typical realization from the proposed null model and to evaluate the power of our test against competing hypotheses. For this reason, we now propose our alternative model based on the Chinese Restaurant Process (CRP). \\

The CRP is a discrete stochastic process used to generate a $\theta-$biased random partition of $\{1,2,...,n\}$ \citep{Aldous1985}. It is used in many Bayesian nonparametric methods, with applications in topic modeling and population genetics \citep{Griffiths2003hierarchical,Qin2006}. Imagine a restaurant with infinitely many tables and $n$ customers who are lined up outside the door in order, with customer 1 first in line. The customers enter the restaurant one at a time. The $k$th customer chooses with probability $\frac{\theta}{k-1+\theta}$ to sit at a new table, and with probability $\frac{1}{k-1+\theta}$ to sit to the left of a particular person already seated. After all the customers have been seated, each non-empty table defines a cycle and the collection of all non-empty tables defines a $\theta$-biased random partition of $\{1,...,n\}$. You would expect more small cycles when $\theta$ is large, and larger cycles when $\theta$ is small.\\

Suppose we are given $N$ samples, $B$ samples of one type (blue), and $N-B$ samples of the other type (red). The CRP-Tree model will generate a random ranked planar partially labeled tree. The parameter in our tree-generating model $\alpha\geq 1$ controls how likely are lineages to descend from a node of the same type. We will construct the tree forward in time, starting at the root.
\begin{enumerate}
    \item Randomly order the $B$ blues and $N-B$ reds into $C=(C_1,...,C_N)$, where $C_i\in\{B,R\}$ is the color label of the $i$th node to be added. 
    \item Form the vector $(w_k= \sum_{i=1}^{k-1} \mathds{1}(C_i=C_k):k=3,...,N)$. Each $w_k$ counts the number of nodes that precede node $k$ that have the same color label as node $k$.
    \item Form a binary tree with two tips, with the left, right tips labeled $C_2, C_1$ respectively.

    \item For $k=3,...,N$: Let $(U_1,...,U_{w_k})$ denote the leaves currently in the tree with same color as node $k$. Let $(V_1,...,V_{k-1-w_k})$ denote the leaves with the opposite color as node $k$. 
    
    \begin{enumerate}
        \item Generate a Bernoulli RV $Z$ with success probability \[ p = \frac{\alpha w_k}{(k-1-w_k)+\alpha w_k} \]
        \item If $Z= 1$, uniformly select leaf $U_i$ from $(U_1,...,U_{w_k})$ to become the parent node of two leaves. Assign label $C_k$ to the left leaf and the label of $U_i$ to the right leaf. 
        \item If $Z=0$, uniformly select leaf $V_i$ from $(V_1,...,V_{k-1-w_k})$ to become the parent node of two leaves. Assign label $C_k$ to the left leaf and the label of $V_i$ to the right leaf. 
    \end{enumerate}
    
    \item After all $N$ tips are added, set the branch lengths so that the length between every consecutive internal node is 1 and all tips are equal distance to the root.
\end{enumerate}

Ordering the tip labels $\{C_1,..., C_N\}$ is equivalent to selecting the sequence of the node colors being added at each step. If $\alpha=1$, then the probability of attaching to any color label is equal. If $\alpha>1$, then the probability of attaching to a node of the same color label is larger. Notice the CRP-Tree model is Markovian because at each stage the transition probabilities only depend on the previous stage.  \\

The analogy to the CRP is as follows. Suppose we have $N$ customers in line and they each have blue or red business cards with the corresponding place in line, such that $B$ customers have blue business cards. The first two customers 1 and 2 walk into the restaurant. If their business cards have the same color, they sit together at the same table, with customer 2 to the left of customer 1. Else, customers 1 and 2 sit at distinct tables. Next, customer $k$ counts $w_k$ customers who have the same color business card as them. With probability $\frac{\alpha}{k-1-w_k + \alpha w_k}$, customer $k$ chooses to sit to the left of a person with the same color business card. With probability $\frac{1}{k-1- w_k +\alpha w_k}$, customer $k$ select a person of the opposite color to ask for their business card. Then customer $k$ moves to a new table, and places the business card to their right. Each customer will sit at only one table, but can have as many business cards at other tables. At the end, each non-empty table with a customer forms an ordered list, and the collection of non-empty tables defines our tree. The order of the tables is irrelevant. Each customer is a tip in our tree, and generating a new table represents a ``mixing'' event of the two colors because two tips of opposite labels are attached together. We give a small example in Figure~\ref{fig:crp_5_tip} showing the tree and corresponding table representation. There will be a smaller number of tables for larger values of $\alpha$, which implies more attachments of the same color.

\begin{figure}[h]
    \centering
    \includegraphics[width=0.75\textwidth]{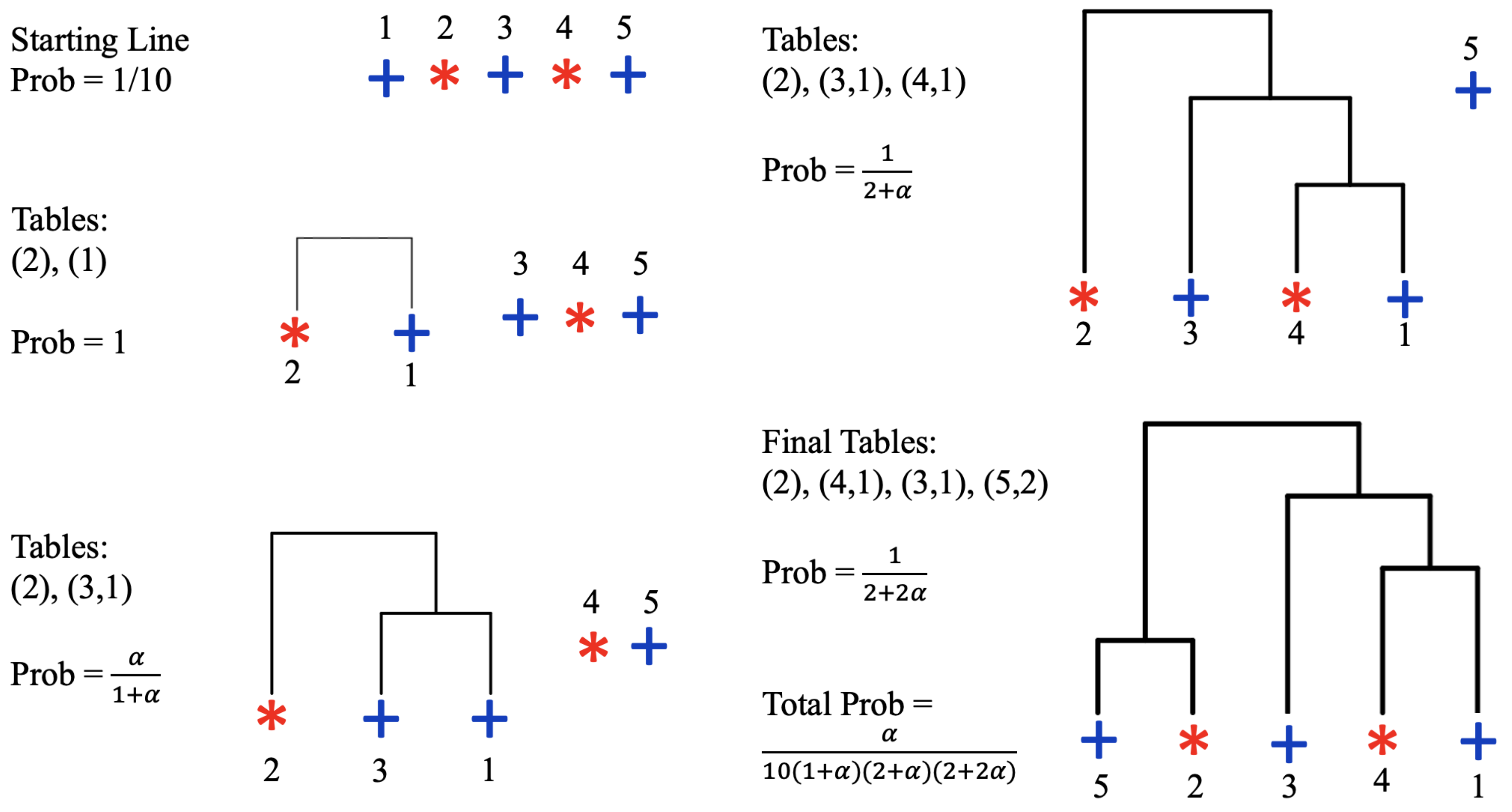}
    \caption{\textbf{An example of a tree generated using the CRP-Tree model with $N=5, B=3$.} Each time we add a node, we show the tree shape, corresponding table configuration, and probability of that attachment. We start with 2 tables because tips 1 and 2 are different colors, and end with 4. There were two same attachments in this example. The rankings at internal nodes are not shown for ease of visualization.}
    \label{fig:crp_5_tip}
\end{figure}

\subsection{CRP-Tree as a 2-urn model}
An intuitive way to understand this generative model is via a 2-urn model. Suppose we start with two urns: Urn 1 has $B$ blue balls and $N-B$ red balls and Urn 2 is empty.
\begin{enumerate}
    \item Select two balls without replacement in order from Urn 1 and place them into Urn 2. Also mark them as Ball 1 and Ball 2. This corresponds to creating a tree with two tips with tip labels $(C_2, C_1)$, colors of Ball $2,1$ respectively.
    \item For $k=3,...,N$:
    \begin{enumerate}
        \item Select 1 ball from Urn 1 and mark it Ball $k$ and note its color $C_k$.  
        \item In Urn 2, assign weight $\alpha$ to balls of the same color as Ball $K$ and weight 1 to the rest of the balls.
        \item Remove a ball from Urn 2 with probability proportional to its weight, call its number $A_k$ and return it to Urn 2.  
        \item In the tree, make $A_k$ the parent node of left leaf with label $C_k$ (the color of Ball $k$) and right leaf with label the color of Ball $A_k$.
    \end{enumerate}
\end{enumerate}

Notice that this implies $W_k$, the number of balls in the first $k-1$ that are the same color as ball $k$, does not depend on $\alpha$. Figure~\ref{fig:urn_ex} shows an example of the 2-urn process for $N=11, B=5, \alpha=2$, after 3 iterations (top panel) and after 4 iterations (bottom panel) when Ball $k=5$ is attached to Ball 2. 

\begin{figure}[h]
    \centering
    \includegraphics[width=0.7\textwidth]{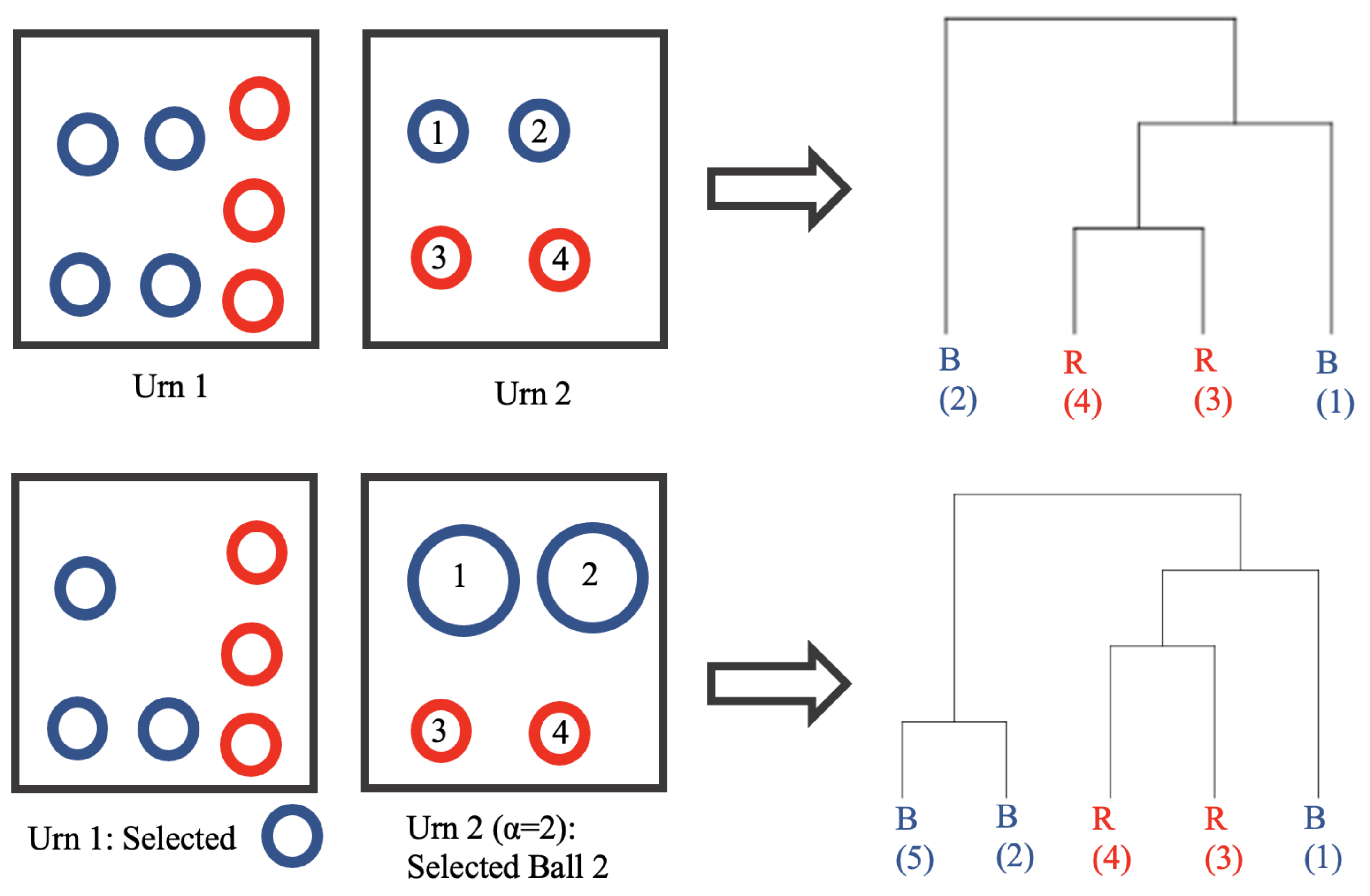}
    \caption{\textbf{An example of two possible states in the urn process corresponding to the CRP-Tree model.} In the top panel, a tree with 4 leaves has already been constructed, with 7 balls in Urn 1 and 4 labeled balls in Urn 2. In the bottom panel, we select one blue ball, increase the size of the blue balls in Urn 2, and ultimately select Ball 2 for Ball 5 to attach to. In the final tree, only the tip colors will be preserved, because the tip numbers can be recovered from the ranking of the internal nodes and the planarity. The rankings at internal nodes are not shown for ease of visualization.}
    \label{fig:urn_ex}
\end{figure}

\subsection{Properties}\label{subsec:properties}
Although the CRP-Tree is well-defined for $B=0, 1, N-1, N$, we will not investigate phylogenetic trait association in these cases. Going forward we will assume that $B \in \{2,3,...,N-2\}$. Let $X_k$ be the indicator of the event that node $k$ attaches to a node of the same color label. Let $S$ be the number of attachments of the same color, that is, $S=\sum_{k=3}^N X_k$. Then the likelihood of the ranked planar partially labeled tree $\tilde{T}^{\ell}_{N,B}$ under the CRP-Tree model is
\begin{align*}
    L(\tilde{T}^{\ell}_{N,B}) &= \frac{1}{\binom{N}{B}} \cdot \prod_{k=3}^N \frac{\alpha^{X_k}}{k-1-w_k+\alpha w_k}\\
    &= \frac{1}{\binom{N}{B}} \cdot \alpha^S \cdot \left (\prod_{k=3}^N (k-1-w_k+\alpha w_k) \right)^{-1}.
\end{align*}
The log-likelihood is \[ \ell(\tilde{T}^{\ell}_{N,B}) = -\log \left (\binom{N}{B} \right) + S \log(\alpha) - \sum_{k=3}^n \log (k-1-w_k+\alpha w_k). \]
We see that $S, \{W_k: k=3,...,N\}$ are the sufficient statistics for $\alpha$ by Fisher–Neyman factorization theorem. Using the CRP-table representation, the number of same tyoe attachments $S$ is given by \[ S = N-2 - \text{\# of new tables}, \]
where we start out with either one table $(2,1)$ in the case $C_{1}$ and $C_{2}$ have the same color, or two tables $(2), (1)$ otherwise. 

\begin{proposition}
If $\alpha=1$, then the probability of any ranked planar partially labeled tree under the CRP-Tree model is \[ \mathbb{P}(\tilde{T}^{\ell}_{N,B})= \frac{1}{(N-1)!} \cdot \frac{1}{\binom{N}{B}}. \]
\end{proposition}
\begin{proof}
We can directly see this result from the likelihood when $\alpha=1$. Alternatively, the probability of any fixed initial ordering is $\frac{1}{\binom{N}{B}}$. Given the initial order, then at step $k$, we uniformly pick a branch to attach to with probability $\frac{1}{k-1}$. After all steps, we get $\frac{1}{(N-1)!}$. Notice this implies the CRP-Tree process uniformly generates ranked planar partially labeled tree shapes under $\alpha =1$ (see Section~\ref{sec:2}).
\end{proof}

\begin{proposition} \label{eq:prob_crp_alpha_1}
If $\alpha=1$, then the probability of a ranked (non-planar) partially labeled tree under the CRP-Tree model is \[ \mathbb{P}(T^{\ell}_{N,B}) = \frac{2^{N-C_S(T^{\ell}_{N,B})-1}}{\binom{N}{B} (N-1)!}. \] 
\end{proposition}
\begin{proof}
There are a total of $2^{N- C_S(T^{\ell}_{N,B}) -1}$ ways to interchanging the left and right subtrees of an internal node without changing the ranked partially labeled tree shape. Combining with Proposition 4 gives the result.
\end{proof}

An important fact implied by Proposition 5 is that the probability of observing $T^{\ell}_{N,B}$ under the CRP-Tree model with $\alpha=1$, is equal to the probability of observing the same tree under the null coalescent model of Section 3. As $\alpha$ increases, the model will generate trees with more same-type attachments. In fact, we will show that as $\alpha$ goes to infinity, the probability of observing a perfect tree (trees with one color completely contained in a monophyletic clade) goes to one. This property will be made formal in the following results. 

\begin{lemma} \label{thm:lemma6}
$\tilde{T}^{\ell, EP}_{N,B}$ is an exact perfect planar tree if and only if $S=N-2$. $\tilde{T}^{\ell, P}_{N,B}$ is perfect but not exactly perfect if and only if $S=N-3$. In addition, under the CRP-Tree model,
\begin{align*}
    \lim_{\alpha\to\infty} \mathbb{P}(\tilde{T}^{\ell, EP}_{N,B}) &= \frac{1}{\binom{N}{B}} \times \frac{1}{(B-1)!(N-B-1)!}, \\
    \lim_{\alpha\to\infty} \mathbb{P}(\tilde{T}^{\ell, P}_{N,B}) &= \frac{1}{(r-1) \binom{N}{B}} \times \frac{1}{(B-1)!(N-B-1)!},
\end{align*}
where $w_r=0$ for a unique $r \in \{3,...,N\}$.
\end{lemma}
\begin{proof}
In order to form an exactly perfect tree, each attachment must be an attachment of the same color, which means $S=N-2$ and we must start with either $B, R$ or $R,B$ in the order of attachments. Hence $w_k\neq 0$ for all $k=3,...,N$, and 
\[ \mathbb{P}(\tilde{T}^{\ell, EP}_{N,B}) = \frac{1}{\binom{N}{B}} \cdot \prod_{k=3}^N \frac{\alpha}{k-1-w_k + \alpha w_k} = \frac{1}{\binom{N}{B}} \cdot \prod_{k=3}^N \frac{1}{w_k + (k-1-w_k)/ \alpha} . \]
Taking the limit gives $\lim_{\alpha\to\infty} \mathbb{P}(\tilde{T}^{\ell, EP}_{N,B}) = \frac{1}{\binom{N}{B}} \cdot \prod_{k=3}^N \frac{1}{w_k}$. Moreover, $\{w_k: k=3,...,N\}$ is an interleaving of $w^B=\{1,2,...,B-1\}$ and $w^R= \{1,2,...,N-B-1\}$, so we have $\prod_{k=3}^N w_k= (B-1)!(N-B-1)!$ and \[ \lim_{\alpha\to\infty} \mathbb{P}(\tilde{T}^{\ell, EP}_{N,B}) = \frac{1}{\binom{N}{B}} \cdot \frac{1}{(B-1)!(N-B-1)!}. \] 

\noindent For perfect but not exactly perfect trees, we must have $S=N-3$ because the root of the monophyletic clade that contains all tips of one color, is attached to an opposite color. This can only happen if the initial order is $B,B$ or $R, R$. Without loss of generality, let us suppose the ordering starts with $B,B$ and that the first $R$ appears at element $r$ ($C_{r}=R$), which would imply $w_r=0$, and this index is unique. In addition, for all $k<r$, we have $w_k=k-1$. Hence, for a perfect $\tilde{T}^{\ell}_{N,B}$ and initial ordering starting with $B,B$, we have 
\begin{align*}
    \lim_{\alpha \to\infty} \mathbb{P}(\tilde{T}^{\ell, P}_{N,B}) &= \lim_{\alpha \to\infty} \frac{1}{\binom{N}{B}} \cdot \frac{\alpha^S}{\prod_{k=3}^N (k-1-w_k + \alpha w_k)}\\
    &= \lim_{\alpha \to\infty} \frac{1}{\binom{N}{B}} \cdot \prod_{k=3}^{r-1} \frac{\alpha}{\alpha w_k} \cdot \left ( \frac{1}{r-1-w_r + \alpha w_r} \right ) \cdot \prod_{k=r+1}^N \frac{\alpha}{(k-1-w_k + w_k \alpha)} \\
    &= \lim_{\alpha \to\infty} \frac{1}{\binom{N}{B}} \cdot \prod_{k=3}^{r-1} \frac{1}{w_k} \cdot \left ( \frac{1}{r-1} \right ) \cdot \prod_{k=r+1}^N \frac{\alpha}{(k-1-w_k + w_k \alpha)} \\
    &= \frac{1}{(r-1)\binom{N}{B}} \prod_{k=3, k\neq r}^N w_k. 
\end{align*}
Now, $\prod_{k=3, k\neq r}^N w_k= (B-1)! (N-B-1)!$ because $\{w_k: k=3,...,r-1, r+1,...N\}$ is an interleaving of $w^B=\{2,...,B-1\}, w^R= \{1,2,...,N-B-1\}$, where $w_r=0$ is not counted in the product. Hence, $\lim_{\alpha\to\infty} \mathbb{P}(\tilde{T}^{\ell, P}_{N,B}) = \frac{1}{(r-1) \binom{N}{B}} \cdot \frac{1}{(B-1)!(N-B-1)!}$. 
\end{proof}

\begin{theorem}
Under the CRP-Tree model with $\alpha \geq 1$, we have
\begin{align*}
    &\lim_{\alpha \to\infty} \mathbb{P}(\{T^{\ell}_{N,B}: T^{\ell}_{N,B} \text{ is perfect}\}) = 1,  \\
    &\lim_{\alpha \to\infty} \mathbb{P}(\{ T^{\ell}_{N,B} : T^{\ell}_{N,B} \text{ is exactly perfect} \} ) = \frac{2B(N-B)}{N(N-1)}, \\
    & \lim_{\alpha \to\infty} \mathbb{P}(\{ T^{\ell}_{N,B} : T^{\ell}_{N,B}\text{ is perfect but not exactly perfect}\}) = 1-\frac{2B(N-B)}{N(N-1)} . 
\end{align*}
\end{theorem}
\begin{proof}
We will first consider a ranked planar partially labeled tree $\tilde{T}^{\ell}_{N,B}$ with $S<N-3$. By the previous lemma, we know that $\tilde{T}^{\ell}_{N,B}$ is not a perfect tree. We will show that the probability of observing such a tree goes to 0 as $\alpha\to\infty$, and therefore, the probability of observing a planar perfect tree goes to 1 as $\alpha \to \infty$. Now, if the initial color ordering starts with $B,B$ or $R,R$, let $r\in \{3,...,N\}$ be such that $w_r=0$. If the initial color ordering starts with $B,R$ or $R,B$, then let $r=2$. We then have 
\begin{align*}
    \lim_{\alpha \to\infty} \mathbb{P}(\tilde{T}^{\ell}_{N,B}) &= \lim_{\alpha \to\infty} \frac{1}{\binom{N}{B}} \cdot \frac{\alpha^S}{\prod_{k=3}^N (k-1-w_k + \alpha w_k)}\\
    &= \lim_{\alpha \to\infty} \frac{1}{(r-1)\binom{N}{B}} \cdot \frac{\alpha^{S}}{\prod_{k=3, k\neq r}^N (k-1-w_k + \alpha w_k)} = 0.
\end{align*}
because the denominator of the second term has leading term $\alpha^{N-3}$ while the numerator has leading term $\alpha^S$. \\

\noindent Any exactly perfect planar tree must be generated with initial color ordering $B,R$ or $R,B$. Given the initial ordering, the number of exactly perfect planar trees that can be formed is $(B-1)!(N-B-1)!$ because at the $k$th step, node $k$ has a choice of $w_k$ nodes to attach to. Moreover, there are $2\binom{N-2}{B-1}$ initial orderings that start with $B,R$ or $R,B$. Therefore by Lemma~\ref{thm:lemma6}, 
\begin{align*}
    &\lim_{\alpha\to\infty}  \mathbb{P}(\{ \tilde{T}^{\ell}_{N,B} : \tilde{T}^{\ell}_{N,B} \text{ is exactly perfect}\}) = \frac{2\binom{N-2}{B-1}}{\binom{N}{B}}= \frac{2B(N-B)}{N(N-1)}, \text{ and } \\
    &\lim_{\alpha\to\infty}  \mathbb{P}(\{ \tilde{T}^{\ell}_{N,B}: \tilde{T}^{\ell}_{N,B} \text{ is perfect but not exactly perfect}\} ) = 1- \frac{2B(N-B)}{N(N-1)} . 
\end{align*}
Finally, note $\mathbb{P}(\{ \tilde{T}^{\ell}_{N,B} : \tilde{T}^{\ell}_{N,B} \text{ is perfect}\}) = \mathbb{P}(\{ T^{\ell}_{N,B} : T^{\ell}_{N,B} \text{ is perfect}\})$ and the same holds for exactly perfect trees. Therefore the three results hold. 
\end{proof}

\begin{theorem}[\textbf{Expected Value of $S$}]~\\
If $\alpha=1$, then \[ \mathbb{E}[S] =\frac{(N-2)\big (B(B-1)+ (N-B)(N-B-1)\big )}{N(N-1)} = (N-2) - \frac{2B(N-B)(N-2)}{N(N-1)}. \]
If $\alpha>1$, then \[ \mathbb{E} [S] = \frac{B}{N} \times \sum_{i=0}^{k-1} \frac{\alpha i}{(k-1-i) + \alpha i} \frac{\binom{k-1}{i} \binom{N-k}{B-(i+1)}}{\binom{N-1}{B-1}} + \frac{N-B}{N}\times \sum_{i=0}^{k-1} \frac{\alpha i}{(k-1-i) + \alpha i} \frac{\binom{k-1}{i} \binom{N-k}{N-B-(i+1)}}{\binom{N-1}{N-B-1}}, \] with the convention $\binom{a}{b}=0$ if $a<b$. 
\end{theorem}
\begin{proof}
Intuitively, $S$ is the sum of linear combinations of Hypergeometric random variables. See Appendix Section~\ref{subsec:expected_value} for the full proof. 
\end{proof}

\subsection{Three equivalent representations}
We now list three ways in which we can represent the information of a ranked planar partially labeled tree shape $\tilde{T}^{\ell}_{N,B}$: the tree form, a sequence of attachments together with an initial color order $C$, and in terms a collection of tables via the CRP. Figure~\ref{fig:tree_tables_list} shows the three representations for $\tilde{T}^{\ell}_{15,6}$. Each representation has its benefits: the tree shape is what is usually given, the sequence of attachments and initial color order allow us to calculate the sufficient statistics $S, \{w_k:k=3,...,N\}$, while the collection of tables is a representation free of any color information. We will show that these representations are all bijective, and describe algorithms to reconstruct each representation from the other. Algorithm 1 allows us to recover the sequence of attachments and initial color order from $\tilde{T}^{\ell}_{N,B}$. Finally, we will define the set of conditions needed for a  collection of tables to encode a ranked planar partially labeled tree, as well as a constructive proof to find the initial color order and the sequence of attachments from the collection of tables. 

\begin{figure}[h]
    \centering
    \includegraphics[width=0.755\linewidth]{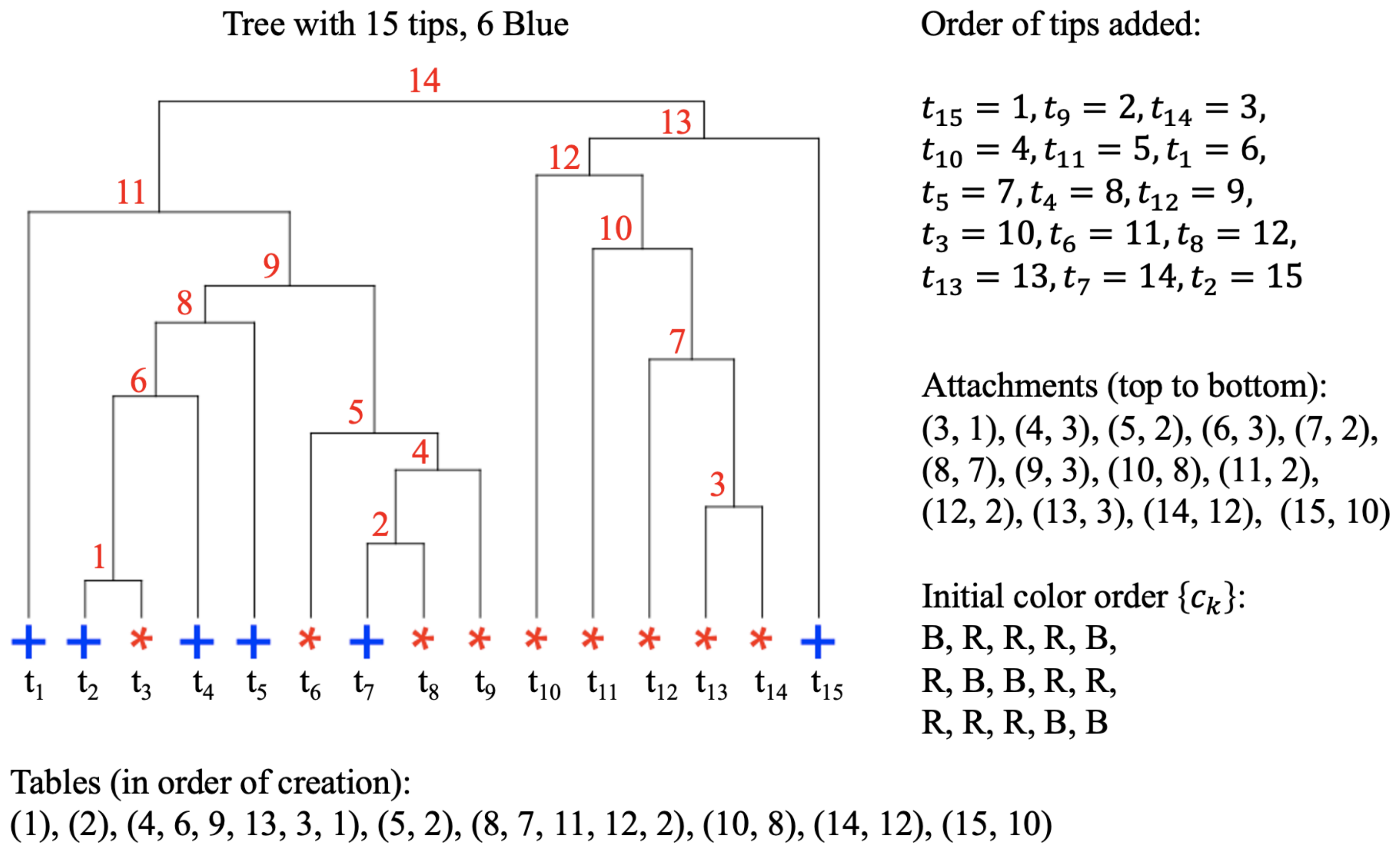}
    \caption{\textbf{The three representations of $\tilde{T}^\ell_{15, 6}$.} First, we label the tips $t_1,..,t_{15}$ from left to right and internal nodes from 1 to 14 from bottom to top (in red). We can determine the order in which the tips were added to the tree and the sequence of attachments from the tree. Similarly, we can derive this same information from the set of tables and vice versa.}
    \label{fig:tree_tables_list}
\end{figure}

\subsubsection{Planar ranked tree to sequence of attachments, initial color order, and order of tips added} \label{subsec:planar-att}
Given $\tilde{T}^{\ell}_{N,B}$, Algorithm 1 in Appendix Section~\ref{subsec:tree_alg} can uniquely determine $C=(C_1,...,C_N)$, the initial ordering of the colors of leaf labels that generated the tree, and the sequence of attachments $\{(k, z_k): k=3,...,N\} $. That is, a ranked planar partially labeled tree has a bijective correspondence with a color ordering $(C_1,...,C_N)$ and a sequence of attachments.\\

The key for being able to recover $C$ and the sequence of attachments from $\tilde{T}^{\ell}_{N,B}$, is our careful planar construction of the tree: new tips are always attached to the left of an existing tip. We first label the internal nodes ranked from bottom to top by 1 to $N-1$, the numbers in red in Figure~\ref{fig:tree_tables_list}. This allows us to uniquely backtrack (bottom to top) which node was being added and which node was begin attached to. We will use the general notation of $t_1,..., t_N$ to label the tips from left to right, where we will find $t_k \in \{1,...,N\}$, the order in which tips were added. \\

We start with the youngest internal node (internal node $1$) and look at its two immediate offspring, the left tip is the last tip added, with color label $(C_N)$ and tip number label $N$, and the right tip is the tip being added to. In Figure~\ref{fig:tree_tables_list}, this corresponds to tip $t_2$ being added and labeled $t_2=15$, $t_3$ being attached to, and $C_N=B$. Proceeding in this manner, the right-most tip in the left subtree of internal node $k$ will be the tip added, while the right-most tip in the right subtree of the internal node $k$ will be the tip added to. 
By nature of the planarity, the tip added would not have been added previously. Returning back to our example, at internal node 2, the tip being added is $t_7$, so it is labeled $t_7=14$ and $C_{14}=B$, and $t_8$ being attached to. For the final step at the root (internal node $N-1$), the right-most tip in the left subtree of the root will be the second tip added, while the right-most tip in the right subtree of the root will the the first tip added. These will be the two remaining tips that have not been added. 
In Figure~\ref{fig:tree_tables_list}, at the root, the right-most tip in the left subtree is $t_9$ and the right-most tip in the right subtree is $t_{15}$, so $t_9=2, t_{15}=1$. One can check that indeed tip $t_9$ had not been added yet. The attachments made are given by the pairs $\{(A_i, D_i)= (N+1-i, t_{r_i}): i=1,...,N-1\}$, with $A_{i}$ denoting the tip being added, and $D_{i}$, the tip being attached to, at step $i$. For example, the last attachment created was $(t_2, t_3)= (15, 10)$. 

\subsubsection{Ranked planar tree to collection of tables and vice versa}
Every ranked planar partially labeled tree $\tilde{T}^{\ell}_{N,B}$ has a one-to-one correspondence to a collection of tables $\{E^1, E^2,...,E^T\}$. These tables obey a set of conditions stated in Appendix Section~\ref{subsec:table_conditions}. An implication of these assumptions is that color information is not needed in the table representation in order to reconstruct the tree. \\

To obtain the sequence of attachments from the collection of tables, first note that the number of attachments represented in a table is the number of elements minus one. For example, the attachments in table $E^t= (E^t_1,...,E^t_n)$, can be denoted by $(E^t_1, E^t_{a_1}), ... (E^t_{n-1}, E^t_{a_{n-1}})$, where $a_j$ is the smallest index satisfying $a_j > j$ and $E^t_j > E^t_{a_j}$. By construction, there will be a total of $N-2$ attachments made across all the tables. We can rearrange all attachments to be in order: $(3, A_3), (4, A_4),..., (N, A_N)$,  including possibly $(2,1)$, with $A_j < j$ for all $j=3,...,N$. For example, the collection of tables $\{E_1, E_2, E_3\} =\{(7, 6,3,1), (5,1), (4,2)\}$ corresponds to attachments $(7,6), (6,3), (3,1), (5,1), (4,2)$.\\

Next, we determine the colors by explicitly stating which attachments must be of the same color. First, if $(2,1)$ appears, then tips $1,2$ must be of the same color, otherwise one is blue and the other is red. If $(2,1)$ does not exist, then the pairs of the form $(j,A_j)$ with $A_j=1$ or $A_j=2$, and $j$ the smallest element that is attached to $A_j$, will be of opposite colors if $(A_j)$ exists as a table in $\{E_1,..., E_T\}$. Otherwise, $j$ and $A_j$ will be of opposite colors. Going back to our example, without loss of generality, we take tip 1 to be Blue and tip 2 to be Red. The smallest tip that is attached to 1 is $3$ in pair $(3,1)$, and $(1)$ does not exist as a table, so tip 3 is also Blue. The smallest element that is attached to 2 is $(4,2)$ and $(2)$ does not exist as a table, so tip 4 must be Red. \\

To determine the remaining attachment types, let $E^{t_j}$ be the table that contains attachment $(j, A_j)$. If $|E^{t_j}|=2$, then $j$ and $A_j$ are of opposite colors. If $|E^{t_j}|\geq 3$, and there is an element smaller than $j$ in the same table, it implies a table was already created when $j$ was added, therefore $j$ and $A_j$ are of the same color. Otherwise, the table is newly created and $j$ and $A_j$ are of different colors. To finish off the example, the table $(5,1)$ has two elements, so tip 5 must be Red. Finally, the attachments $(7,6)$ and $(6,3)$ are attachments made on pre-existing tables, so tips $3, 6, 7$ are all the same color. Therefore, our final set of attachments and tip colors is $(7B, 6B), (6B,3B), (3B,1B), (5R,1B), (4R,2R)$. 

\subsection{Resulting Tree Topology}\label{subsec:tree_topology}
In Section~\ref{subsec:properties}, we showed that the CRP-Tree model with $\alpha=1$ generates ranked partially labeled trees with the same probability law as the null coalescent model of Section~\ref{sec:3}. We further showed that if we remove the color labels, we obtain unlabeled ranked tree shapes with the same law as in the Tajima coalescent. We empirically verify that $\alpha$ does not greatly affect the probability law of the ranked tree shapes by comparing the averages of various tree statistics to the expectation under the Tajima coalescent. Two popular statistics are the number of cherries (2-tip subtrees), and the number of pitchforks (3-tip subtrees), with expected values $N/3$ and $N/6$ respectively under the standard coalescent \citep{McKenzie2000, Choi2020}. Figure~\ref{fig:cherry_pitchfork} in the Appendix shows that for each $\alpha$, the number of cherries and pitchforks is concentrated around the expected value. \\

Another simulation to examine the distribution of ranked tree shape topologies is motivated by \cite{Kim2020}. The authors proposed a distance on ranked tree shapes and use it to visualize tree distribution in 2 dimensions via multidimensional scaling (MDS). Figure~\ref{fig:mds_plot} in the Appendix does not exhibit clustering of the trees per distribution, suggesting similarity among the three ranked tree shape distributions. More details on both these simulations can be found in Appendix Section~\ref{subsec:tree_topology_details}. 

\subsection{Discussion about planarity}
The output of our model is a ranked planar partially labeled tree, which is a largely unexplored tree resolution. In this case, planarity and the tip colors allows us to label the internal nodes blue or red and to know the sequence of node attachments. In many biological situations, there is a differentiation between the two children of a node, and therefore a way to distinguish them is by keeping track of left and right nodes. For example, these include speciation, transmission trees in epidemiology, and cell lineage diagrams,, where in each case, the left and right subtrees represent a biologically important distinction \citep{Stewart2005, Hagen2015, Sainudiin2016b}. \\

Other methodologies are also built under the assumption of planarity. \cite{Behr2020} assume a fixed planar representation, and they show via simulation that their method is mostly robust to changes in this representation (i.e. their  statistic does not change much with a change in the planarity). \cite{Ford2009} also work on the resolution of ranked planar trees. Their test statistic is calculated on the planar tree, and then planarity is marginalized out when calculating the final p-value. \cite{Sainudiin2016} derive a Beta-splitting model at a variety of tree resolutions: ranked and planar, unranked and planar, ranked and non-planar, unranked and non-planar. 

\section{CRP-Tree phylogenetic association test}\label{sec:5}
We will first consider the setting in which a ranked partially labeled phylogenetic tree is available, for example obtained via Maximum Likelihood estimation from molecular sequence data. We ignore any phylogenetic uncertainty in this case, which may not be ideal in many situations. In addition, we will assume that a binary trait is completely observed at the tips of the phylogeny, that is, we observe $T^{\ell,0}_{N,B}$. We use the superscript 0 to denote that the tree is observed and fixed. Given $T^{\ell,0}_{N,B}$, we wish to test whether there is a phylogenetic association in the binary trait. In terms of the CRP-Tree model, the hypothesis of no phylogenetic association is equivalent to $H_{0}:\alpha=1$.\\

A natural test statistic is $S$, the number of same type attachments. However, this statistic is not directly observed since it depends on the initial color ordering and the sequence of type of attachments. Instead, our proposed test statistic is $\mu=\mathbb{E}[S|T^{\ell}_{N,B} ]$. In practice for large $N$, $\mu$ is replaced by  \[ \hat{\mu} = \frac{1}{M} \sum_{i=1}^M S_i, \]  where $S_{i}$ is the number of same-type attachments in $\tilde{T}^{\ell, i}_{N,B}\sim \mathbb{P}_{\alpha=1}(\tilde{T}^{\ell, i}_{N,B}\mid T^{\ell}_{N,B})$.  

\subsection{Testing by permutation}\label{subsec:pval-def}
To assess significance, we estimate the null distribution of our test statistic $\hat{\mu}$, conditional on the observed ranked tree shape $T_{N}: T^{\ell,0}_{N,B} \prec T_{N}$, and $(N,B)$. We estimate the null distribution by random permutation of the leaf labels, that is, we generate $\{T^{\ell,i}_{N,B}\}^{K}_{i=1}$ by randomly relabeling the tips of $T_N$ with $B$ blues and $N-B$ reds, $K$ times. To calculate $\hat{\mu}_{0},\hat{\mu}_{1},\ldots,\hat{\mu}_{K}$, we sample $M$ ranked planar partially labeled tree shapes uniformly conditional on each $T^{\ell,i}_{N,B}$, for $i=0,\ldots,K$ by picking an internal node uniformly among those that do not subtend a same type cherry, and then permuting its left and right subtrees. The $p$-value is then \[ p_S = \frac{1+\sum_{i=1}^K \mathds{1}(\hat{\mu}_i \geq \hat{\mu}_0)}{1+K} . \]
When $K$ and $M$ are larger than the possible number of permutations, we simply generate all permutations to compute the exact $p$-value. Figure \ref{fig:permutation_schematic}(a, c) show a schematic of the simulations needed to calculate $p_{S}$.\\

We propose a second test in which our test statistic is a sample of $\{S^{obs}_{i}\}_{i=1}^{M}$ from $\mathbb{P}_{\alpha=1}(S\mid  T^{\ell,0}_{N,B})$ generated by sampling planar representations uniformly for the observed $T^{\ell,0}_{N,B}$, and computing the number of same type attachments. The cardinality of the space of all ranked planar partially labeled trees $\tilde{T}^{\ell}_{N,B}$ that are compatible with $T_N$ is $2^{N-1-C(T^{\ell}_{N,B})}\times \binom{N}{B}$. Therefore, our null distribution is generated by randomly sampling $K$ partial labelings on $T_{N}$ together with a random planar representation. Let $\{S_j\}_{j=1}^K$ denote the empirical null distribution. Then, our p-value is defined as \[ p_T = \frac{1}{M} \sum_{i=1}^M \left ( \frac{1+ \sum_{j=1}^K \mathds{1}(S_j \geq S^{obs}_i)}{1+K} \right ) \] Figure~\ref{fig:permutation_schematic}(a, b) show a schematic of the simulations needed to calculate $p_{T}$.

\begin{figure}[h]
    \centering
    \includegraphics[width=0.9\textwidth]{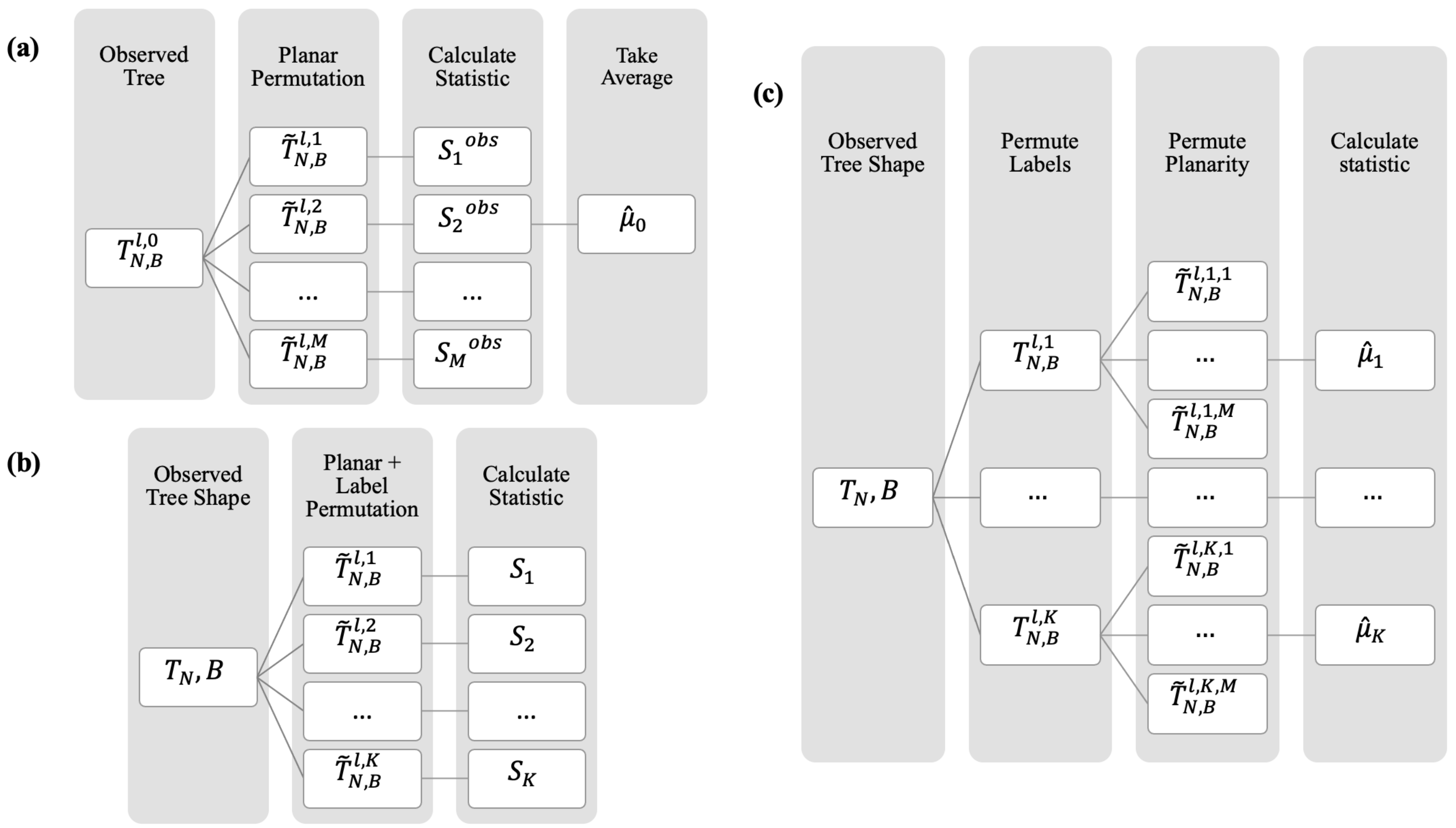}
    \caption{\textbf{A schematic for calculation of p-values.} In (a), we obtain $\hat{\mu}_0$ as the average of $\{S^{obs}_{i}\}^{M}_{i=1}$ statistics obtained by permutation of the planarity of $T_{N,B}^{\ell,0}$. In (b), we obtain an empirical null distribution of $S$ given $T_N, B$ by permuting planarity and tip labels. In (c), we generate a sample of $\{\hat{\mu}_{i}\}^{K}_{i=1}$ from the null distribution given $T_N, B$. The sample in (a) is used in both computation of $p_S, p_T$, the samples in (b) are needed to compute $p_T$, and the samples in (c) are needed to compute $p_{S}$.}
    \label{fig:permutation_schematic}
\end{figure}

\begin{lemma}\label{lemma:P_S}
 Let $T^{\ell,0}_{N,B},\ldots,T^{\ell,K}_{N,B}\overset{\text{iid}}{\sim} \mathbb{P}(T^{\ell}_{N,B}| T_N, B)$, where $\mathbb{P}(T^{\ell}_{N,B}| T_N)$ is the distribution of ranked and partially labeled tree shapes $T^{\ell}_{N,B} \prec T_{N}$ derived in Proposition~\ref{eq:null_model_cond_prob}. Let $\mu_{i}=\mu(T^{\ell,i}_{N,B})=\mathbb{E}[S\mid T^{\ell,i}_{N,B}]$ for $i=0,\ldots,K$, and let \[ P_S = \frac{\sum_{i=0}^K \mathds{1}(\mu_i \geq \mu_0)}{1+K},\] then $\mathbb{P}_{\alpha=1}(P_{S}\leq \alpha)\leq \alpha$ for all $\alpha \in [0,1].$
\end{lemma}
\begin{proof}
First we condition on $\mu_0, T_N$. 
\begin{align*}
    \mathbb{P}_{\alpha=1} ( P_S \leq \alpha \mid \mu_0, T_N ) &= \mathbb{E}_{\alpha=1} \left [ \mathds{1}\left \{ \frac{1}{K+1} \sum_{i=0}^K \mathds{1} (\mu_i \geq \mu_0) \leq \alpha \right \} \mid \mu_0, T_N \right ] \\
    &\leq \mathbb{E}_{\alpha=1} \left [  \mathds{1} \left \{ \sum_{T_{N,B}^{\ell,i}} \mathbb{P}(T^{\ell, i}_{N,B}| T_N) \;\; \mathds{1} \big (\mu(T_{N,B}^{\ell,i}\big ) \geq \mu_0) \leq \alpha \right \} \mid \mu_0, T_N  \right ], 
\end{align*}
since the sample proportion converges to the population proportion by the Law of Large Numbers as $K \rightarrow \infty$, the last inequality results from Portmanteau's Lemma \citep{van2000asymptotic}. Then 
\begin{align*}
    \mathbb{P}_{\alpha=1} (P_S \leq \alpha \mid T_N ) &\leq \mathbb{E}_{\alpha=1}\left [ \mathbb{E}_{\alpha=1} \left \{ \mathds{1} \left ( \sum_{T_{N,B}^{\ell,i}} \mathbb{P}(T^{\ell, i}_{N,B}| T_N) \;\; \mathds{1} \big (\mu(T_{N,B}^{\ell,i}\big ) \geq \mu_0) \leq \alpha \right ) \mid \mu_0, T_N  \right \} \right ]  \\
    &= \sum_{T_{N,B}^{\ell,j}} \mathbb{P}(T^{\ell, j}_{N,B}| T_N) \;\; \mathds{1} \left ( \sum_{T_{N,B}^{\ell,i}} \mathbb{P}(T^{\ell, i}_{N,B}| T_N) \;\; \mathds{1} \big (\mu(T_{N,B}^{\ell,i}) \geq \mu(T_{N,B}^{\ell,j}) \big ) \leq \alpha \right)
\end{align*}
Lemma A1 of \cite{harrison2012conservative} states for all $t_0,...,t_n \in [-\infty, \infty], \alpha, w_0,...,w_n \in [0,\infty]$, then \[ \sum_{k=0}^n w_k \mathds{1}\left (\sum_{i=0}^n w_i \mathds{1}(t_i\geq t_k) \geq \alpha \right ) \leq \alpha. \]
From this result, we deduce $\mathbb{P}_{\alpha=1} ( P_S \leq \alpha \mid T_N ) \leq \alpha$ and therefore $\mathbb{P}_{\alpha=1}(P_{S}\leq \alpha) \leq \alpha$. 
\end{proof}

\begin{theorem}
$p_S$ and $p_T$ are asymptotically valid p-values. 
\end{theorem}
\begin{proof}
To show the result for $p_S$, notice that $\hat{\mu}_i \xrightarrow{p} \mu_{i}=\mu(T^{\ell,i}_{N,B})$ by the Law of Large Numbers as $M\rightarrow \infty$. Therefore, by the Continuous Mapping Theorem, $\hat{\mu}_i - \hat{\mu}_0 \xrightarrow{p} \mu_i-\mu_0$. By Lemma~\ref{lemma:indicator_conv} in Appendix Section~\ref{subsec:lemma_conv}, for all $i=1,...,K$, \[\mathds{1}(\hat{\mu}_i - \hat{\mu}_0 \geq 0 ) \xrightarrow{p}\mathds{1}(\mu_i - \mu_0 \geq 0 ),\] Again applying Continuous Mapping Theorem, $p_S \xrightarrow{p} P_S$ as $M\rightarrow \infty$. Hence, $\lim_{M \rightarrow \infty} \mathbb{P}_{\alpha=1}(p_S \leq \alpha) = \mathbb{P}_{\alpha=1}(P_S \leq \alpha) \leq \alpha$  for all $\alpha\in [0,1]$, by Lemma~\ref{lemma:P_S}.  \\ 

\noindent To show the result for $p_T$, first note that with probability 1, \[ p_T \geq \frac{1}{M} \sum_{i=1}^M \frac{1}{K} \sum_{j=1}^K \mathds{1}(S_j \geq S^{obs}_i) :=p_T' \] and therefore it suffices to show $p_T'$ is a valid p-value. Switching the order of summation gives $p_T' = \frac{1}{K} \sum_{j=1}^K \frac{1}{M} \sum_{i=1}^M \mathds{1}(S_j \geq S_i^{obs})$. First consider $S_j$ as fixed. As $M$ increases \[ \frac{1}{M} \sum_{i=1}^M \mathds{1}(S_j \geq S_i^{obs}) \xrightarrow{a.s.} \mathbb{P}_{\alpha=1}(S_j \geq S^0 \mid T_{N,B}^{\ell, 0} ):= f(S_j) \] where $S^0 \sim \mathbb{P}_{\alpha=1}(S \mid T_{N,B}^{\ell, 0} )$. Now $\frac{1}{K}\sum_{j=1}^K f(S_j)$ is a permutation p-value, which implies $p_T$ is an asymptotically valid p-value.
\end{proof}

In our simulations, we take $K=M$ between 200 and 500. We choose to use a permutation test instead of a likelihood-ratio test because we do not have an analytical expression for the likelihood of a ranked (non-planar) partially labeled tree shape under the alternative.

\subsection{Testing in the Bayesian framework}
In the previous section we assumed that a ranked and partially labeled phylogeny was observed, however, phylogenies are typically not directly observed. Here we consider the case when one would use BEAST, or other Bayesian inference implementation, to generate a posterior distribution of the trees given molecular sequence data \citep{suchard2018bayesian, ronquist2012mrbayes}. The tip label information (i.e. colors) is not used to generate these posterior trees. To account for phylogenetic uncertainty, we propose to simply estimate the posterior distribution of p-values $p_T$ or $p_S$, and reject the null hypothesis according to whether there is posterior evidence of the p-values being smaller than the significance value. \\

We compare our method to BaTS (Bayesian Tip-association Significant testing) proposed by \cite{Parker2008}, where a test statistic is obtained for each tree in the posterior sample, and the posterior median $m_0$ is used as the test statistic. Next, $n$ random permutations of the color labels $\sigma_{1},\ldots,\sigma_{n}$ are generated such that all trees in the posterior distribution are relabeled according to the same permutation. From each permutation, a median statistic is obtained to generate a null posterior distribution of the median test statistic.  The p-value is obtained by calculating the proportion of $m_i$ values that are more extreme than the observed $m_0$. Notice that BaTS is effectively ignoring the Bayesian uncertainty by using the posterior median as the test statistic of interest, and therefore may yield small credible intervals. We will investigate this further in Section~\ref{sec:6}. 

\subsection{Power of the test by MCMC}
To estimate the power of the test under different alternatives, we approximate the distribution of $\hat{\mu}$ or $S$ under our model with $\alpha \neq 1$, conditional on the observed ranked tree shape $T_{N}$. We approximate this distribution via Metropolis-Hastings \citep{Hastings1970}. Given current state $\tilde{T}^{\ell, y}_{N,B}$, our proposal distribution generates $\tilde{T}^{\ell, x}_{N,B}$ by first assigning tip labels uniformly and then a node is chosen uniformly at random to swap its left and right subtrees, among those that are not cherries of the same type. This proposal is symmetric and therefore, our acceptance probability is simply $r= \frac{\mathbb{P}_\alpha(\tilde{T}^{\ell, x}_{N,B})}{\mathbb{P}_\alpha(\tilde{T}^{\ell, y}_{N,B})}$. Uniformly sampling the labels is more efficient than a local label move given that we only have two unique tip labels. Since all proposals are generated conditioning on a given tree shape, the stationary distribution of the Markov chain is $\mathbb{P}_\alpha(\tilde{T}^{\ell, i}_{N,B} \mid T_{N},B)$. In practice, we found that generating $M$ Metropolis-Hastings steps of planarity swaps per one step of relabeling of the tips, improves the mixing of the chain considerably.

\section{Simulation Results}\label{sec:6}

We first use our method to test for phylogenetic trait association for the tree in Figure \ref{fig:trait_ex}(b) with a test statistic value of $\hat{\mu}_0=10.27$. Here we assumed a number of $M=K=500$ of planar and label permutations to obtain the two p-values $p_S=0.719, p_T= 0.746$. The two plots in Figure~\ref{fig:hist_tree_25} show the approximate null distributions of $\hat{\mu}_0$ and $S$ conditioned on $T_N$ and $B$. In this case, we have strong support for no phylogenetic association, that is, we would not reject the null hypothesis of no association at the $5\%$ significance level. \\

\begin{figure}[h]
    \centering
    \includegraphics[width=0.95\textwidth]{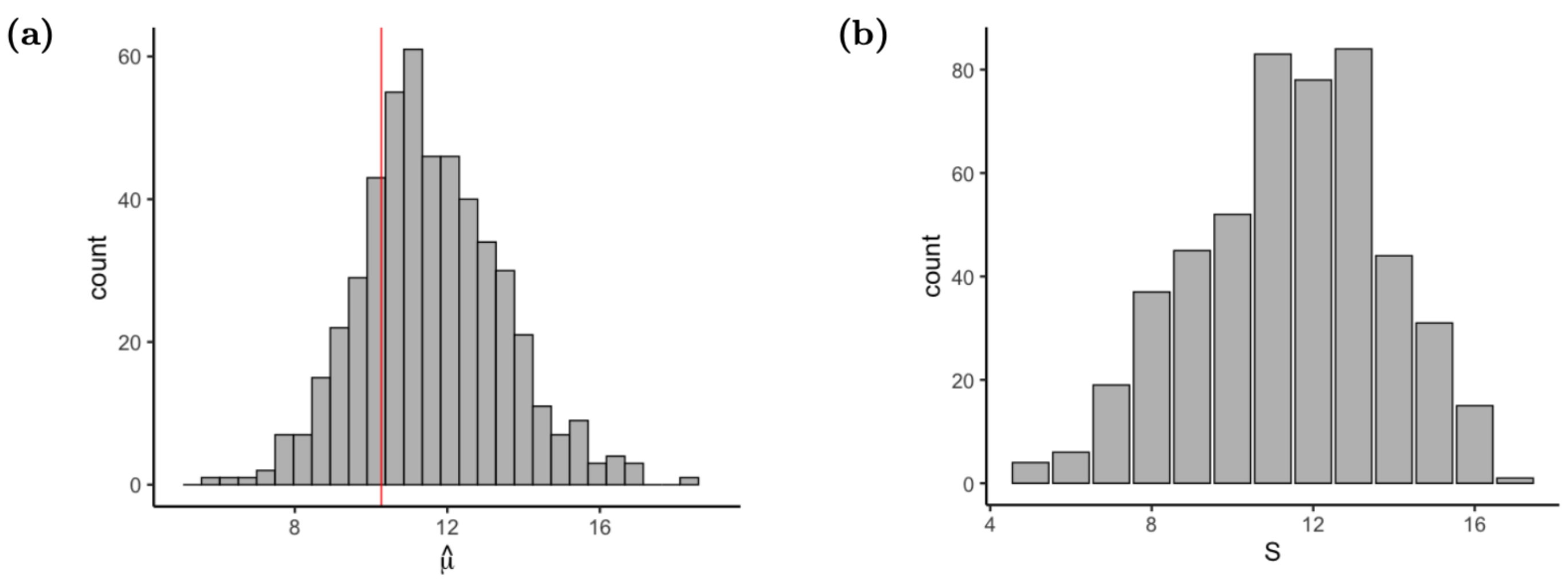}
    \caption{\textbf{Analysis of the fixed tree $T^{\ell}_{25,15}$ in Figure~\ref{fig:trait_ex}(b).} (a): Approximate null distribution of $\hat{\mu}$, with observed value in red. (b): Approximate null distribution of $S$.}
    \label{fig:hist_tree_25}
\end{figure}

In order to compare the performance of our two proposed tests, we simulated 200 random ranked partially labeled tree shapes according to the CRP-Tree with the following parameter values: $N\in \{20, 50, 100, 200, 500\}$, $\alpha \in \{1, 2, 5, 10, 25\}$, and fraction of one the types $B/N \in \{0.1, 0.25, 0.5\}$. For the results depicted in Table~\ref{table:alpha_testing_results} we assume a significance level of 5\%, $M=K=200$ planar and label permutations. In general, both methods concur (97.7\% of the time when $\alpha=1$, and 87.2\% of the time when $\alpha>1$). When the trees are simulated under $\alpha=1$ (left table), $p_{S}$ is correctly above the significance threshold only few more times than $p_{T}$ ($99.8\%$ vs $97.5\%$). When the trees are simulated under $\alpha>1$, $p_{T}$ is correctly below the significance $75.4\%$ of the time, and $p_{S}$ $62.5\%$ of the time. We conclude that $p_S$ is more conservative than $p_T$ and that both $p_T$ and $p_S$ control the Type 1 error rate. 

\begin{table}[h]
\centering
\begin{tabular}{|rrrr|}
     \hline
    $\alpha=1$ &  & $p_S$ & \\ \hline
    &  & Do not reject & Reject \\
    $p_T$ & Do not reject & 2925 &   0 \\ 
     & Reject &  69 &   6 \\ 
   \hline
\end{tabular}
\hspace{0.05\linewidth}
\begin{tabular}{|rrrr|}
     \hline
    $\alpha>1$& & $p_S$ & \\ \hline
    &   & Do not reject & Reject \\
    $p_T$ & Do not reject & 2951 &   0 \\ 
     & Reject &  1540 &   7509 \\ 
   \hline
\end{tabular}
\caption{\textbf{Contingency tables comparing $p_{T}$ and $p_{S}$}. Each entry represents the number of times the two methods rejected or not when the true value is $\alpha=1$ or $\alpha>1$. Simulations carried out across a range of $N$ and $B$ values.}
\label{table:alpha_testing_results}
\end{table}

\subsection{Power analyses}\label{subsec:power_sim}
We investigate the power of our tests and compare them to that of Parsimony score (PS), Association Index (AI) and treeSeg via simulation. The Parsimony Score (PS) \citep{Fitch1971} counts the minimum number of state changes in the phylogeny in order to reconstruct the states at the parent nodes. Here, the states at the tips are binary and so we use the Fitch algorithm for the score computation. Small values imply strong phylogenetic trait association. The Association Index (AI) \citep{Wang2001} is defined by $AI = \sum_{i=1}^{N-1} \frac{1-f_i}{2^{m_i-1}}$, where $m_i$ is the number of tips subtended by internal node $i$ and $f_i$ is the frequency of the most common trait value among the tips subtended. Note that smaller value of AI implies stronger phylogenetic trait association, because the numerator $1-f_i$ is smaller for larger $f_i$. We also apply the changepoint detection method implemented in treeSeg \citep{Behr2020} to the sample of trees. In this case, detection of at least one changepoint corresponds to rejecting the null hypothesis. \\

We first generated the five different ranked tree shapes depicted in Figures~\ref{fig:tree_25_random_power}-\ref{fig:tree_100_unb_power}. Three of these tree shapes were generated uniformly at random $(N=25, 50, 100)$ while the other two tree shapes are the most balanced and most unbalanced trees $(N=100)$ according to the criteria defined in \citet{rajanala2021statistical}. For each ranked tree shape, we approximated the power of the test under 4 different frequencies of the two types: $B/N \in \{0.1,0.25,0.4, 0.5\}$, and under 4 alternatives: $\alpha \in \{2,5,10,20\}$. We set $M=300$ for calculating $\hat{\mu}_{0}$ and $\{S^{obs}_{i}\}^{M}_{i=1}$ and generated $500$ planar and label MCMC steps. \\

The power approximations are displayed in the tables of Figures~\ref{fig:tree_25_random_power}-\ref{fig:tree_100_unb_power} in the Appendix. We see that our two methods have much better power than any of the pre-existing statistics AI and PS, as well as treeSeg in all cases. We generally observe increasing power as $B/N$ and $\alpha$ increase for each fixed tree. We note that we do not require a very large tree in order to be able to detect phylogenetic trait association. However, we do notice that less balanced tree shapes and an imbalance of label types may result in lower power. \\

We extend our power study to 100 randomly simulated ranked tree shapes with $N=50$. The boxplots for the power of each test under $B/N \in \{0.1, 0.5\} $ and $\alpha\in \{2,5,10,20\}$ are depicted in Figure~\ref{fig:power_boxplot}. We confirm that our methods consistently perform better than AI and PS, with the $\hat{\mu}$ statistic achieving the highest power overall. We do not compare to treeSeg due to its high computational time. 

\begin{figure}[h]
  \centering
   \includegraphics[width=0.75\linewidth]{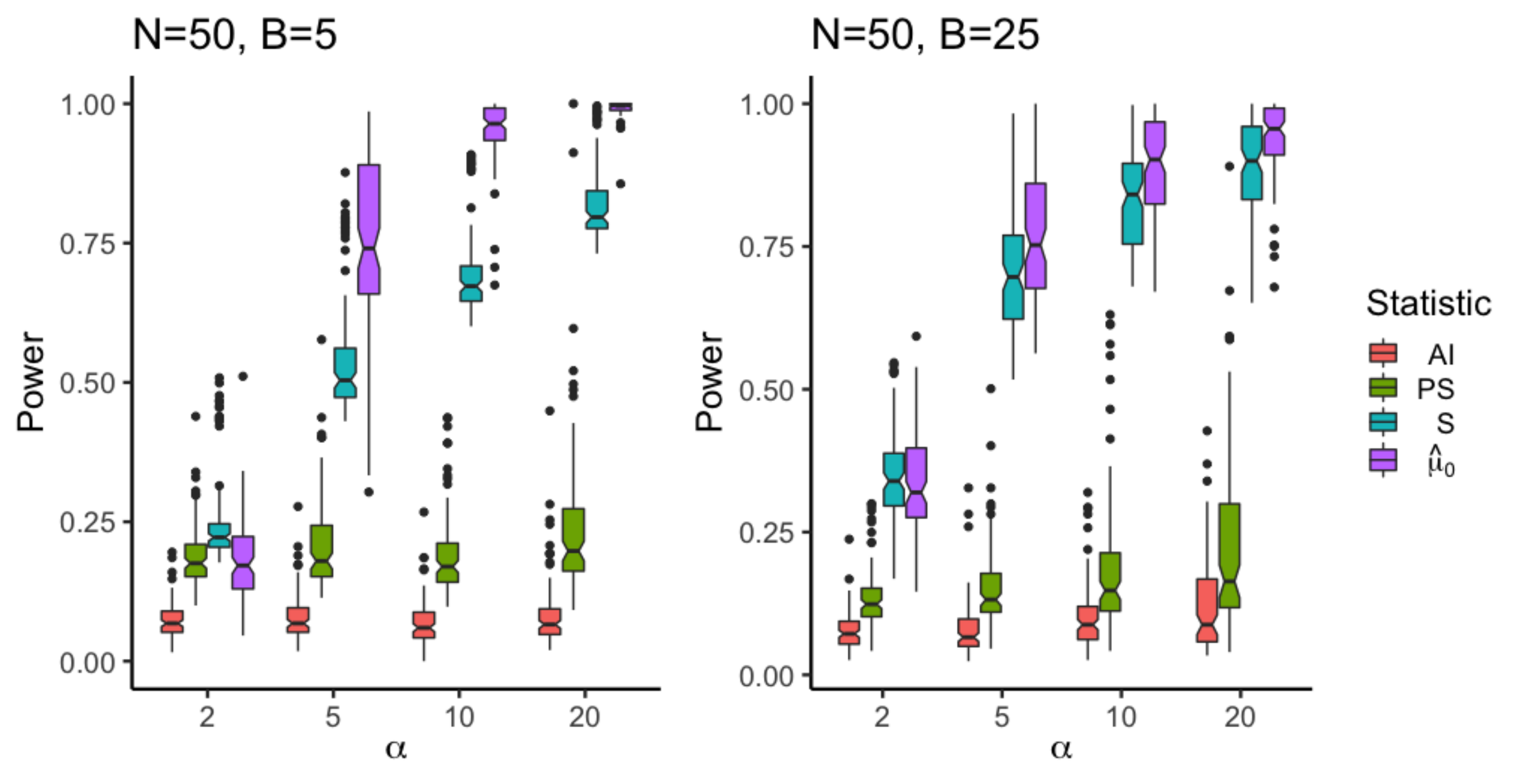}
   \caption{\textbf{Power simulations for 100 tree shapes with $N=50$.} The left plot shows the power for $B=5$ and the right plot shows the power for $B=25$.}
   \label{fig:power_boxplot}
\end{figure}

\subsection{Posterior Validation of p-values}\label{subsec:bayesian_sim}
As a validation check in the Bayesian setting, we first simulated two phylogenies ($N=50$, $B=20$) from the CRP-Tree model with $\alpha=1$ and $\alpha=10$ respectively, and simulated DNA sequences at the tips of each phylogeny. We then used BEAST \citep{suchard2018bayesian} to estimate the two posterior distributions and tested the null hypothesis of $\alpha=1$ in both cases. Details on simulation experiment can be found in Appendix Section~\ref{subsec:dna}. We compared the posterior distributions of the p-values obtained with our test to the p-values obtained with BaTS. When $\alpha=1$, the posterior mean p-value obtained with our method is $0.86$, and the posterior median is $0.866$. When $\alpha=10$, the posterior mean p-value is $0.0021$ and the posterior median is $0.0021$. The posterior distributions of the $p$-values are depicted in Figure~\ref{fig:bayesian_simulation}. In both cases, the user would have correctly concluded the true association. However, the BaTS p-values are both 0, which implies it would incorrectly reject the first case of $\alpha=1$. We will show in the next section that we usually obtain concordant conclusions from the posterior distribution of p-values and BaTS p-values.

\begin{figure}[h]
    \centering
    \includegraphics[width=0.95\textwidth]{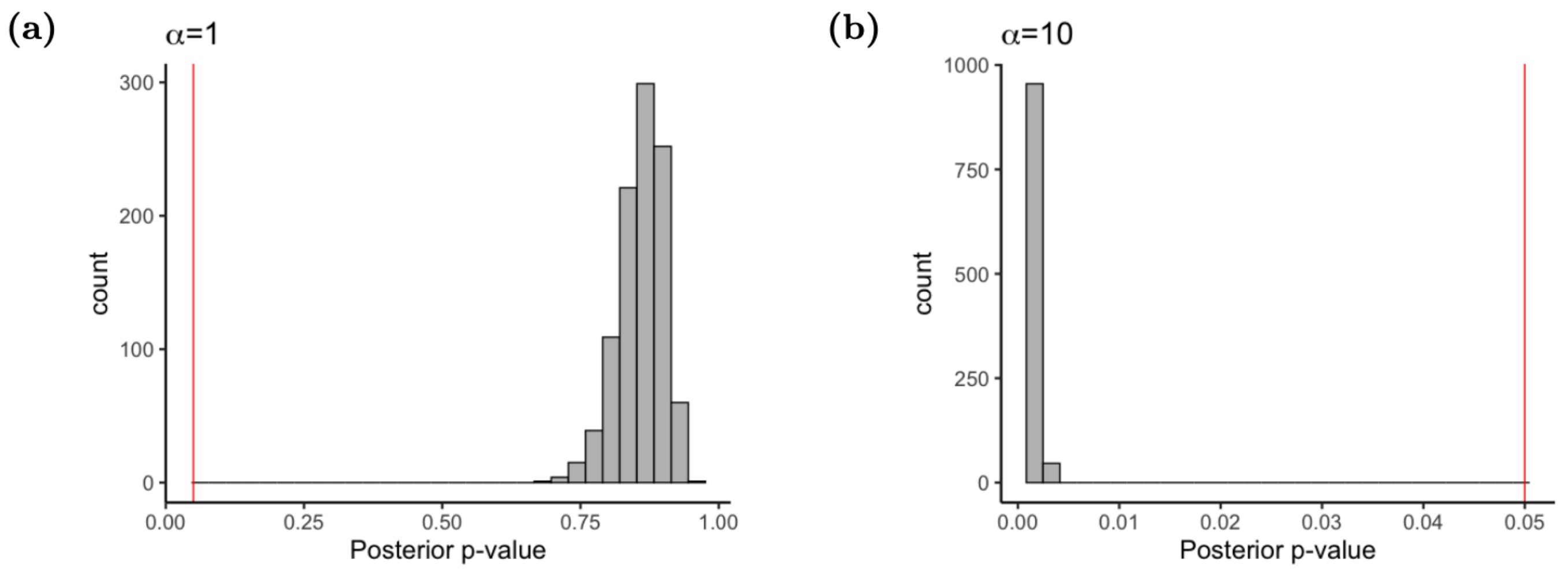}
    \caption{\textbf{Posterior distributions of p-values.} Posterior distribution of $p_T$ for the simulation with (a) $\alpha=1$, and (b) $\alpha=10$. Red line is marks p-value of 0.05.}
    \label{fig:bayesian_simulation}
\end{figure}

\section{Case studies}\label{sec:7}
We first apply our tests to two real data studies in which the ranked partially labeled tree shapes are available (without known uncertainty). We then apply our test to two studies in which the posterior distribution of trees is estimated via MCMC from molecular sequence data.

\subsection{A breast cancer study}
We re-analyze a publicly available breast cancer gene expression study from 98 patients \citep{Vantveer2002} in which more than five thousand genes were found to be significantly associated to breast cancer, out of a pool of approximately 25 thousand genes. An additional six clinical responses were collected: BRCA mutation, estrogen receptor expression, histological grade, lymphocytic infiltration, angioinvasion, and development of distant metastasis within 5 years, although this last variable had missing data and is excluded in this study. \cite{Behr2020} apply a hierarchical clustering algorithm using a similarity metric on these regulatory genes to create a tree. The results of our tests for tree association to each of these clinical responses are shown in Figure~\ref{fig:behr_ex}. The only trait that is not rejected with our methods is the association of the angioinvasion trait with the tree structure. These results are consistent with \cite{Behr2020}. 

\begin{figure}[h]
    \centering
    \includegraphics[width=0.6\textwidth]{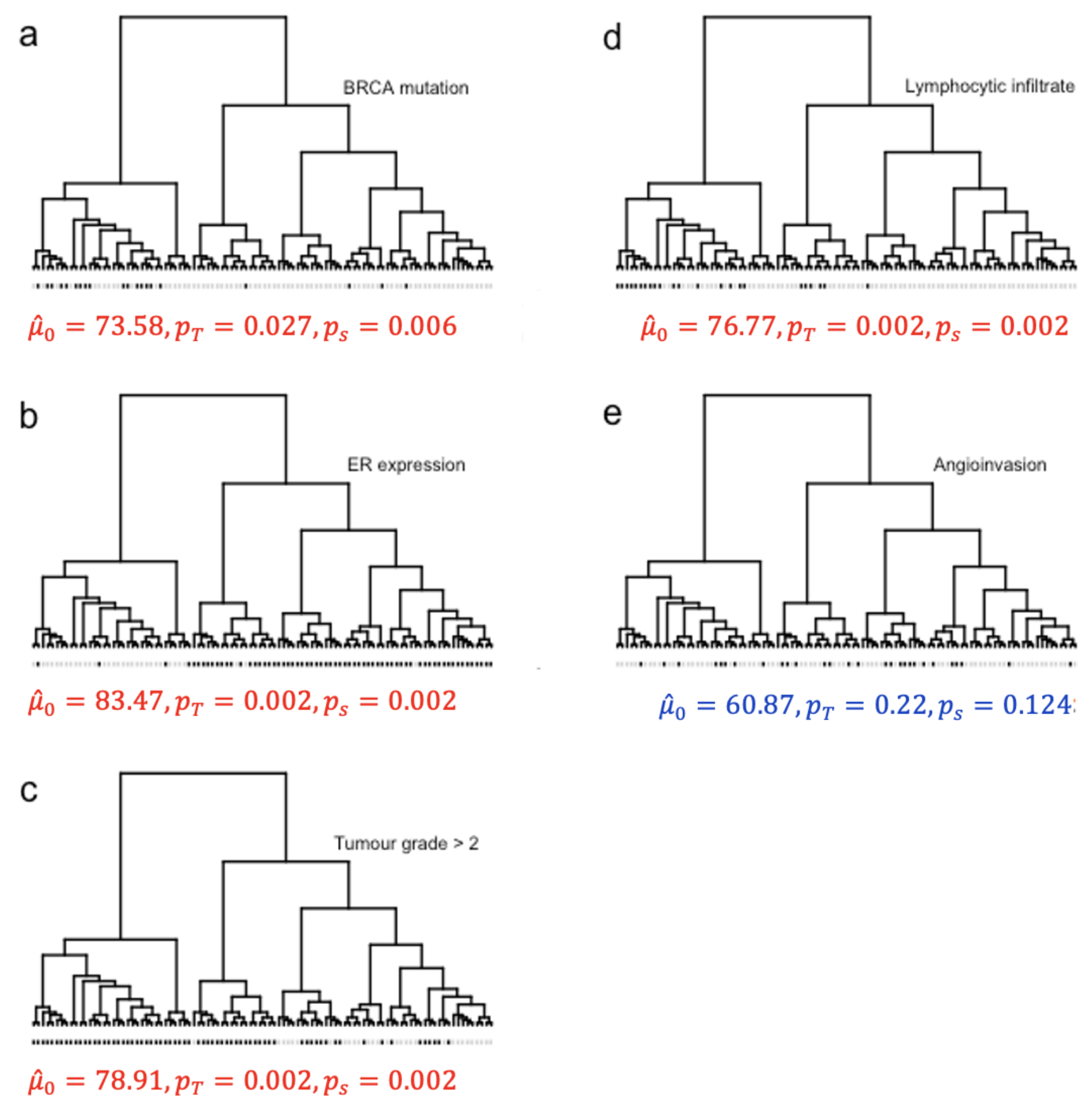}
    \caption{\textbf{Phylogenetic association tests of five clinical responses associated to breast cancer $(N=98)$.} The tick marks indicate whether the subject had the trait or not. The angioinvasion trait in (e) is the only case where we would not reject the null hypothesis with our tests.}
    \label{fig:behr_ex}
\end{figure}

\subsection{Sexually attractive traits in Swordtail fish}
In evolutionary biology, swordtail fish (\textit{Xiphophorus}) are a classic model for studying sexual selection \citep{darwin1871}. Decades of research have shown that females have preferences for large male body size \citep{ryan1987asymmetries, Rosenthal1998, cummings2006} and several sexually selected ornaments, including the ``sword'' ornament for which the genus is named \citep{rosenthal2001shared, basolo2002}. \cite{Preising2022} collected wild-caught individuals and used whole genome sequencing to infer phylogenetic relationships and examine the co-evolution of certain traits within the \textit{Xiphophorus} clade. Their phylogenetic tree, constructed via maximum likelihood, is shown in Figure~\ref{fig:fish_tree_original}. Here, our interest is to test for phylogenetic association of the two main traits: presence/absence of the sword, and whether the size of the body is larger than 26.5 inches. \\ 

\begin{figure}[h]
    \centering
    \includegraphics[width=0.65\linewidth]{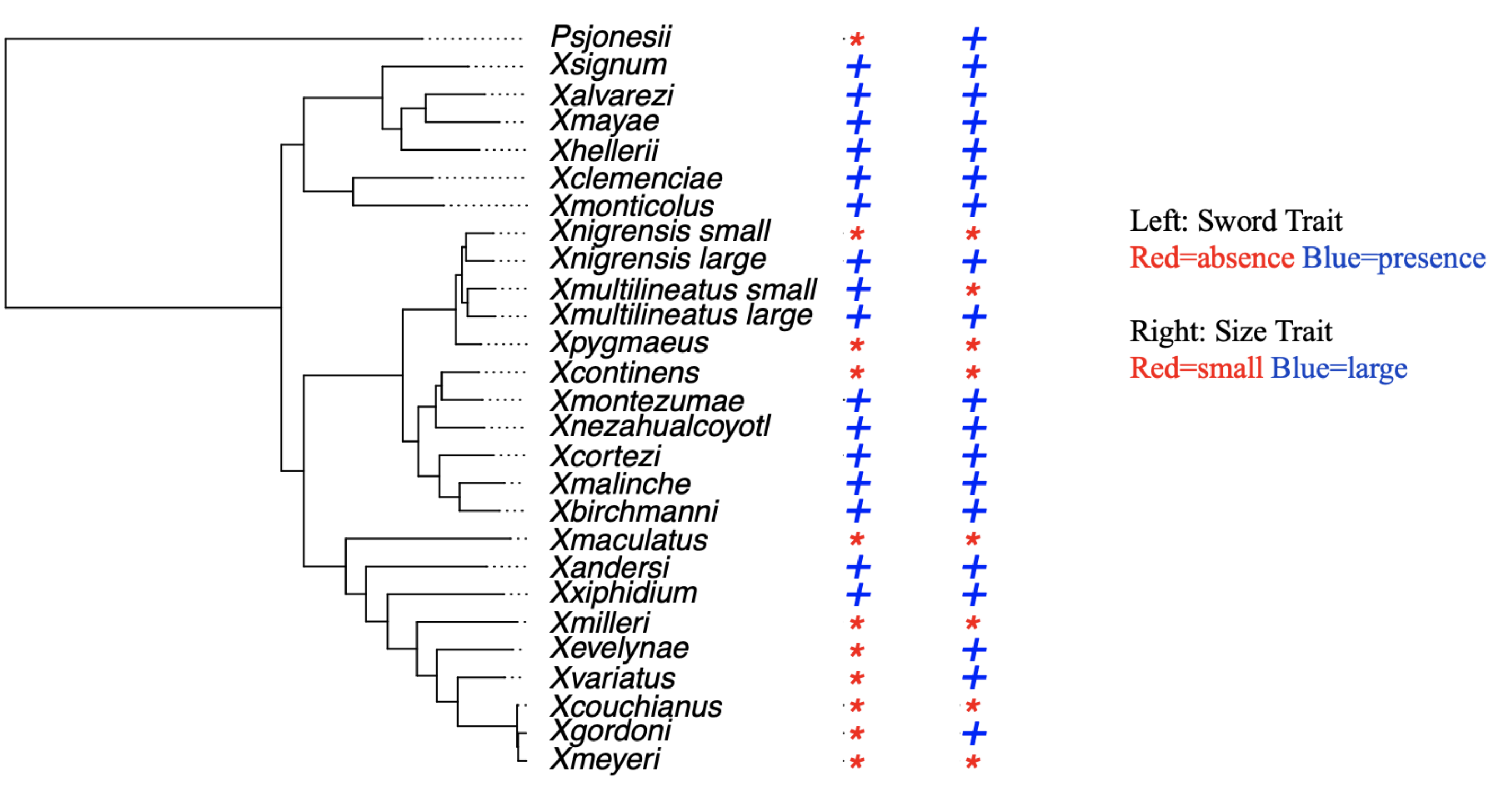}
    \caption{\textbf{A phylogeny of 27 swordtail fish species and the two traits of interest.} For the sword trait, $B=16$ and for the size trait, $B=19$.}
    \label{fig:fish_tree_original}
\end{figure}

We generated 500 planar permutations and 500 label permutations to approximate the null distribution of $\hat{\mu}$ for the two traits (Figure~\ref{fig:fish_hists}). For the size trait, we obtained $\hat{\mu}_0=13.3$, $p_T=0.714$, and $p_S= 0.686$, and so we conclude that the size of the fish is not associated with the phylogeny. This is somewhat expected since the phylogeny is based on the whole genome and body size is associated with a small number of polymorphic sites \citep{lampert2010determination}. For the sword trait, we obtained $\hat{\mu}_0=17.05, p_T= 0.0706, p_S= 0.018$ and so the presence/absence of the sword appears to be associated with the tree topology. 

\begin{figure}[h]
    \centering
  \includegraphics[width=0.95\textwidth]{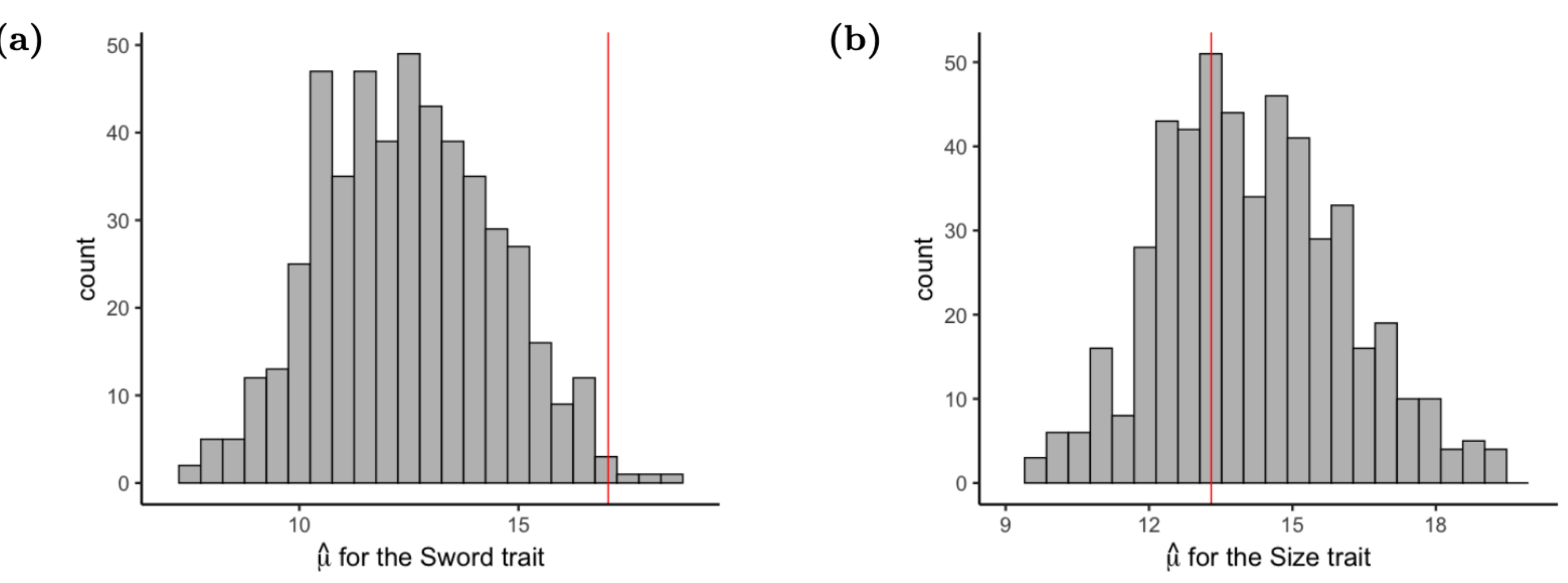}
    \caption{\textbf{Analysis of the swordtail fish sequences.} Null Distribution of $\hat{\mu}_0$ for the presence/absence of Sword trait (a) and large/small Size trait (b). Red line is the observed value.}
    \label{fig:fish_hists}
\end{figure}

\subsection{Transmission cycles of Brazilian yellow fever virus}
Brazil recently experienced yellow fever virus (YFV) outbreaks in multiple states causing more than 700 deaths between 2016 and 2018. YFV is the most severe mosquito-borne infection in South America. Non-human primates are usually infected by \textit{Haemagogus} spp. and \textit{Sabethes} spp., while humans are infected by \textit{Aedes} spp. In \citep{Faria2018}, the authors generated 65 complete viral genomes collected from 33 infected humans and 32 non-human primates across several states in Brazil during 2016-2017. In order to investigate whether the virus is spreading between human and nonhuman primates, the authors estimated the movement of YFV lineages between humans and nonhuman primates using a structured coalescent model. The authors estimated a variable rate of transmission from nonhuman primates to humans rising from zero around November 2016 and reaching a peak in February 2017. Here, we reanalyze the same sequences in order to investigate whether the virus is spreading between humans and nonhuman primates. If the virus was spreading between the two populations, there should not be a phylogenetic association with the human-nonhuman trait.\\

To incorporate phylogenetic uncertainty, we obtained a sample of 5,000 phylogenetic trees (after thinning every 100,000 iterations) from the posterior distribution generated with the software BEAST \citep{suchard2018bayesian}. The posterior distribution of phylogenies is completely agnostic to the human-nonhuman label. Details of model assumptions can be found Appendix Section~\ref{subsec:dna}. We calculate $\hat{\mu}_0$ and p-values $p_T$ for each tree in the posterior sample. We generated 500 permutations of the planarity to generate each $\hat{\mu}_0$ and 500 permutations of the labels to simulate each null distribution conditional on each ranked tree shape in the posterior. The posterior distribution of the p-values shown in Figure~\ref{fig:yfw_post_figure}(a) has a mean of 0.04 and median of 0.031, with 70\% of the values below the 0.05 significance level. This result suggests that the virus spreads mainly within each population. This new result contradicts the original finding of variable migration from nonhuman to human populations. However, the authors in the original study alert caution that hypotheses of human-to-human transmission should not be tested directly using phylogenetic data alone, due to large undersampling of NHP infections.

\begin{figure}[h]
    \centering
   \includegraphics[width=0.95\textwidth]{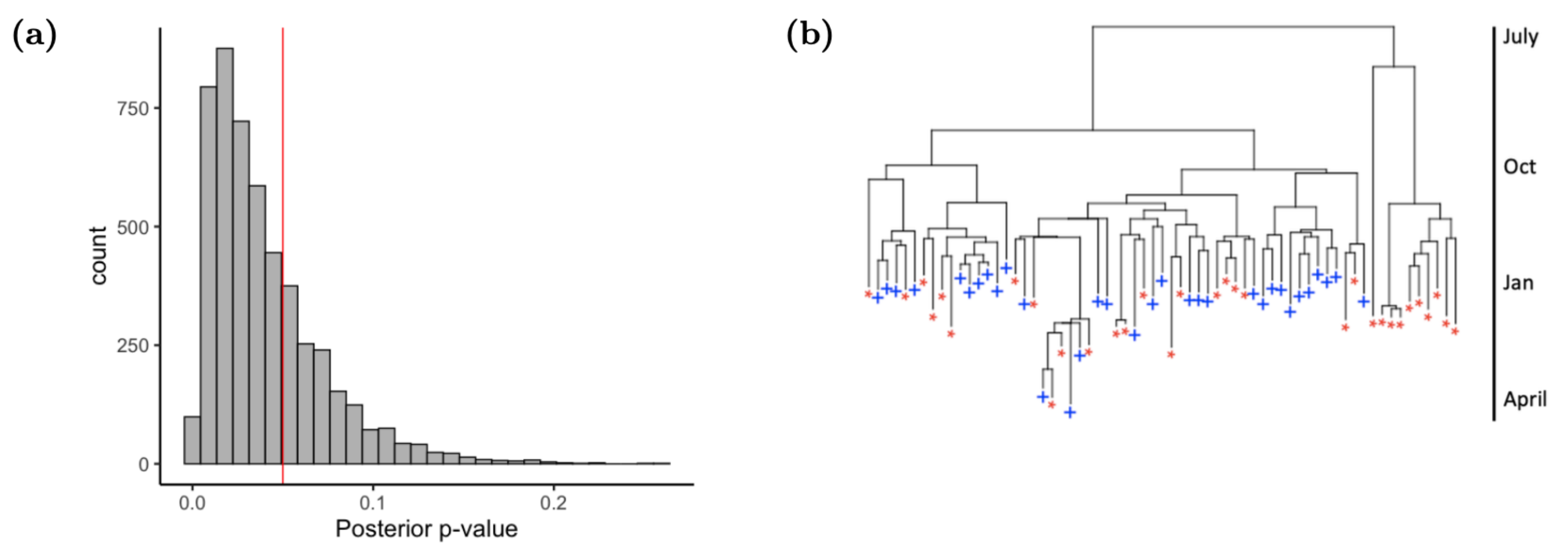}
    \caption{\textbf{Analysis of Yellow Fever Virus}. (a): Posterior distribution of p-values $p_T$ for the Yellow fever virus study in Brazil. Red line is placed at p-value of 0.05. (b): Maximum-clade credibility tree, with branch lengths indicative of the date of collected samples in 2017. Blue tips correspond to human samples and red to nonhuman primate samples.}
    \label{fig:yfw_post_figure}
\end{figure}

The observed posterior median value of $\hat{\mu}_0$ is 39.15. When we apply BaTS to our test statistic with 500 permutations of the label set, the 95\% credible interval obtained for the medians of $\hat{\mu}_0$ under the null is $[30.49, 31.48]$ which allows us to obtain the same conclusion with our method: observed data suggests some level of preferential attachment. \\

Finally, we obtained a single tree from the posterior distribution that corresponds to the maximum clade credibility tree in Figure~\ref{fig:yfw_post_figure}(b) and performed our test on this tree ignoring phylogenetic uncertainty. In this case, $\hat{\mu}_0= 42.62, p_T = 0.0066$, in agreement with our previous result. 

\subsection{Population structure in H1N1 Transmission}
In early 2009, the swine-origin influenza A (H1N1) virus originated as a novel combination of influenza genes. In just one year, the estimated number of H1N1 cases arose to more than 60 million in 2010 \citep{CDC2009}. \cite{Smith2009} study the origins and evolutionary genomics of the H1N1 pandemic using phylogenetic analyses on related virus genomes, such as H3N2, classical swine H1N1, and North American avian. \cite{suchard2018bayesian} analyze 50 H1N1 viral genome sequences, a subset of the original study to estimate the origin date of the pandemic, the growth, and basic reproductive number. A sample of 1,000 posterior trees was generated in BEAST (thinning every 10,000 states) without using the geographical location of the sequences. We use our test to determine if there is population structure in the transmission of H1N1, a question that was not investigated in the original study. \\

Since geographic location is not binary, we choose to split our data into USA (22 sequences) and non-USA (28 sequences). Using the same procedures as the previous example, we generated $M=500$ permutations of the planarity to generate each $\hat{\mu}_{0}$ and $K=500$ permutations of the labels to simulate each null distribution conditional on each ranked tree shape in the posterior. The posterior distribution of the p-values (Figure~\ref{fig:h1n1_usa_post_figure}(a)) has a mean of 0.0525 and median of 0.042, with close to 60\% of the values below the 0.05 significance level. Therefore, there seems to exist population structure in the spread of influenza between strands in the USA and otherwise. The observed posterior median value of $\hat{\mu}_0$ is 31.32. Using BaTS with 500 label permutations, the 95\% credible interval obtained for the medians of $\hat{\mu}_0$ under the null is $[23.43, 24.33]$ which allows us to obtain the same conclusion with our method: observed data suggests some level of preferential attachment.

\begin{figure}[h]
    \centering
   \includegraphics[width=0.95\textwidth]{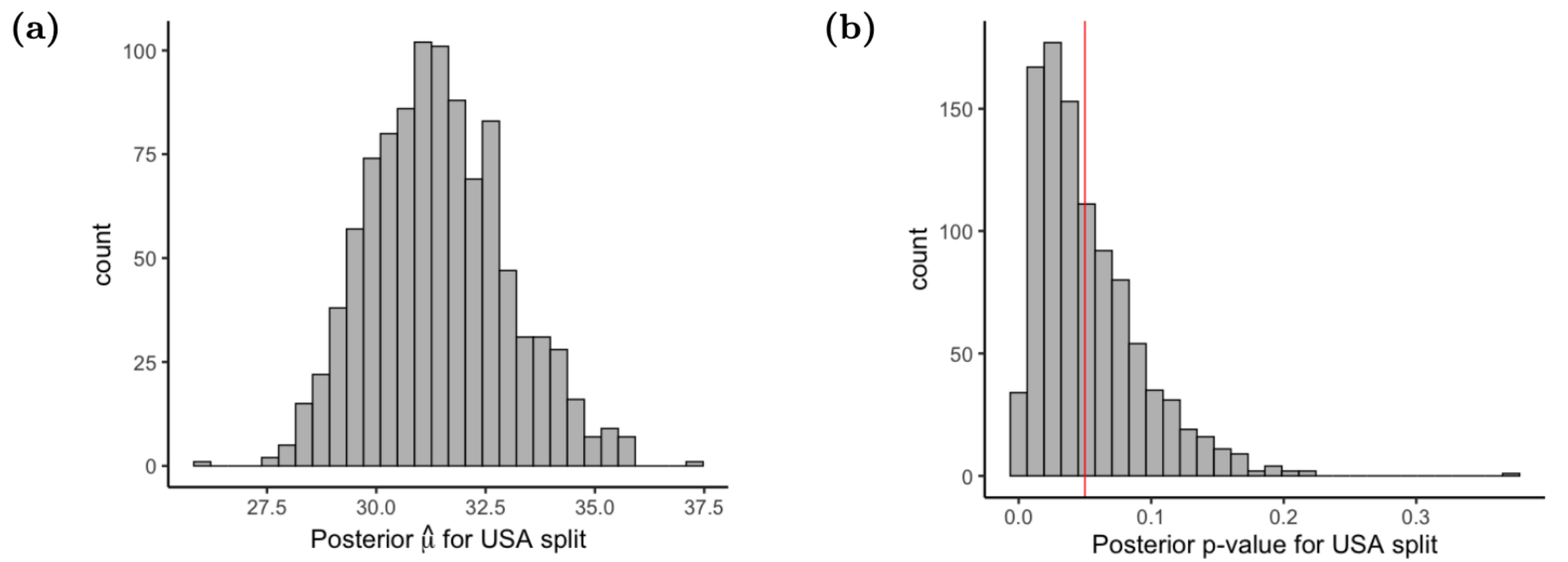}
    \caption{\textbf{Analysis of H1N1 Virus with USA/non-USA trait split}. Posterior distribution of $\hat{\mu}$ (a) and posterior p-values $p_T$ (b) calculated on a posterior sample of 1000 trees. Red line is 0.05.}
    \label{fig:h1n1_usa_post_figure}
\end{figure}

\section{Discussion and Extensions}\label{sec:8}

We propose a nonparametric permutation test of phylogenetic association with a binary trait. We empirically demonstrate that our test is more powerful than the parsimony score, the association index and treeSeg. We also showed that our test is computationally efficient. In average, testing takes about 7 seconds for a tree with 100 tips and 40 seconds for a tree with 500 tips. In particular, our test is much faster than treeSeg, which could not produce results for larger trees. In addition, we extend our test to the setting when the phylogeny is not directly observed and instead, a posterior distribution of phylogenetic trees is available. \\ 

The proposed test statistic, the mean number of same type attachments in a tree, is a sufficient statistic for the single parameter of the CRP-Tree model proposed in this manuscript. In the CRP-Tree model, this parameter controls the likelihood of lineages to attach to lineages of the same type and hence, the test statistic arises in a natural and interpretable way. However, our proposed CRP-Tree model cannot be used directly in a likelihood ratio test formulation. The difficulty lies in the fact that the CRP-Tree model is a model on planar and ranked trees, a finer tree resolution than the one usually observed (non-planar and ranked trees). Nevertheless, our proposed statistic is amenable to a permutation test, and in fact, the permutation p-value corresponds to the Monte Carlo p-value obtained by sampling planar ranked and partially labeled tree shapes with the same ranked tree shape as the observed tree. 
Although our permutation p-value is a special case to the ones described in \citep{hemerik2018exact,ramdas2022permutation}, the validity of our p-values relies on the fact that randomly permuting leaf labels and planarity, effectively generates i.i.d. samples from the conditional CRP-Tree model with $\alpha=1$ (null model). We note that our test statistic, the average number of same type attachments, is closely related to the parsimony score, the minimum possible value of different-type attachments. However, our statistic showed better power and overall performance in all scenarios considered.

\subsection{Possible extensions}
Our method is applicable to categorical traits with more than two possible values. Under the CPR-Tree model, a node attaches to a node of the same type or to a node of a different type, according to the same probabilities derived for the binary case. Therefore, the test statistic remains the same. However, in our experience, there usually is a large loss of power when the number of categories increases and therefore, we do not pursue it here.\\

Another possible extension is to allow for missing tip information. An effective way would consist in replacing our statistic, to the expected number of same type attachments conditional only on the partial labels, effectively, integrating out uncertainty in tip labels. Although we do not actively pursue this extension here, we do not anticipate any difficulty in its implementation. \\

We can also consider an extension of the CRP-Tree model to continuous trait values. The notion of preferentially attaching to a node of the same color is replaced by attaching to tips with smaller absolute difference. When $\alpha \to\infty$, new tips would always attach to the tip with the smallest absolute difference. Letting $Y_i$ denote the trait value for the $i$th tip being added, we define the test statistic to be $S= \sum_{i=1}^N |Y_i - Y_{d_i}|$, where $(i, d_i)$ is an attachment. If there is preferential attachment, then $S$ will be small. Notice that taking $Y_i \in \{0,1\}$ for binary traits gives $S$ equal to the number of same attachments. Again, permutation tests can be utilized for assessing whether there is any association or not. This bypasses the need to assume any parametric model of trait evolution, such as Brownian Motion \citep{Pagel1999, Blomberg2003}, or the Ornstein–Uhlenbeck process \citep{felsenstein1988phylogenies, butler2004phylogenetic}. \\

We are acutely aware that our test ignores branch length information. We could extend our model by allowing $\alpha$ to vary with time, where larger $\alpha$ would imply shorter branch lengths for an attachment of the same type. If $\alpha$ were allowed to vary with time, say according to a Poisson process, then we could test hypotheses on the subtrees sequentially to discover the location of these changepoints. That is, a subtree can be considered as a CRP-Tree realization with smaller values of $N,B$. In addition, we could place priors on $\alpha(t)$ and use Bayesian nonparametric methods for inferring $\alpha(t)$. \\

Finally, we note that method is not suitable for complex traits in which multiple genes are involved. An extension of our method for networks could be possible and subject of future research.

\newpage

\bibliography{main}

\newpage 

\section{Appendix}\label{sec:9}

\subsection{Expected Value of $S$}\label{subsec:expected_value}
Let $X_k$ be the indicator of the event that node $k$ attaches a node of the same color label, and let $\mathcal{C}=\{C_{1},\ldots,C_{N}\}$ be an ordering of attachments. Then, the sequence $\{W_k=w_{k}\}$ is known. We know $X_k \mid  \mathcal{C} \sim Bern( \frac{\alpha w_k}{(k-1-w_k) + \alpha w_k})$ and $X_k, X_j$ are independent for $k\neq j$. Therefore,
\begin{align*}
    \mathbb{E}[S| \mathcal{C}] &= \sum_{k=3}^N \frac{\alpha w_k}{(k-1-w_k) + \alpha w_k}, \\
    \mathbb{E}[S] &= \mathbb{E}\big [ \mathbb{E}[S | \mathcal{C}] \big ] = \sum_{k=3}^N\mathbb{E}\left [ \frac{\alpha W_k}{(k-1-W_k) + \alpha W_k} \right ] .
\end{align*}

In order to calculate $\mathbb{E}[S]$, we only need the distribution of $W_k$, which does not depend on $\alpha$. If we choose $k$ balls from Urn 1 without replacement, the number of balls that match the color of the $k$th ball out of the first $k-1$ chosen is $W_k$. Let $Y_k^{(B,N)}$ be the number of blue balls drawn after drawing $k$ balls without replacement from an urn with $B$ blue balls and $N-B$ red balls. Then $Y_k^{(B,N)} \sim$ Hypergeometric$(B,N, k)$, that is, for $i\in \{\max(0, k-(N-B)),..., \min(B, k)\}$: \[ \mathbb{P}(Y_k^{(B,N)}=i) = \frac{\binom{B}{i} \binom{N-B}{k-i}}{\binom{N}{k}}. \] 
Now, given that $k$ balls have already been drawn from our pool of $B$ blue and $N-B$ red balls, let $Z_k^{(B,N)}$ be the indicator that the $k$th drawn ball is blue. Then,  
\begin{align*}
    \mathbb{P}(W_k=i) &= \mathbb{P}(Y_{k-1}=i, Z_k=1) + \mathbb{P}(Y_{k-1}=k-1-i, Z_k=0) \\
    &= \mathbb{P}(Z_k=1| Y_{k-1}=i) \mathbb{P}(Y_{k-1}=i)+ \mathbb{P}(Z_k=0| Y_{k-1}=k-1-i)\mathbb{P}(Y_{k-1}=k-1-i) \\
    &= \frac{B-i}{N-(k-1)} \frac{\binom{B}{i} \binom{N-B}{k-1-i}}{\binom{N}{k-1}} + \frac{N-B-i}{N-(k-1)} \frac{\binom{B}{k-1-i} \binom{N-B}{i}}{\binom{N}{k-1}} \\
    &= \frac{B}{N} \times\frac{\binom{B-1}{i} \binom{N-B}{k-1-i}}{\binom{N-1}{k-1}} + \frac{N-B}{N} \times \frac{\binom{B}{k-1-i} \binom{N-B-1}{i}}{\binom{N-1}{k-1}} .
\end{align*}
We recognize that this probability is 
\begin{equation} \label{eq:pmf_W_k}
    \mathbb{P}( W_k=i) = \frac{B}{N} \times \mathbb{P}\left ( Y_{k-1}^{(B-1, N-1)}=i\right ) + \frac{N-B}{N}\times  \mathbb{P}\left (Y_{k-1}^{(N-B-1, N-1)}=i\right ) .
\end{equation}
The interpretation of the two Hypergeometric RVs is that $Y_{k-1}^{(B-1, N-1)}$ denotes the number of blue balls in $k-1$ drawn balls from $N-1$ total balls and $B-1$ blues and $Y_{k-1}^{(N-B-1,N-1)}$ denotes the number of red balls in choosing $k-1$ balls from $N-1$ total balls and $N-B-1$ red. Therefore, $W_k$ is a convex combination of two independent Hypergeometric random variables.\\

\noindent Expanding out the previous expression for $\mathbb{P}(W_k=i)$, we get 
\begin{align}
    \mathbb{P}(W_k=i) &= \frac{B! (N-B)! (k-1)!(N-k)!}{i! (B-i-1)! (k-1-i)!(N-B-(k-1-i))! N!} + \nonumber \\
    &\qquad \frac{B! (N-B)! (k-1)!(N-k)!}{i! (N-B-i-1)! (k-1-i)!(B-(k-1-i))! N!} \nonumber \\
    &= \left [ \frac{1}{(B-i-1)! (N-B-(k-1-i))!} + \frac{1}{(N-B-i-1)! (B-(k-1-i))!}  \right ] \times \nonumber \\
    &\qquad \frac{B! (N-B)! (k-1)!(N-k)!}{i! (k-1-i)! N!}\nonumber \\
    &= \frac{\binom{k-1}{i}}{\binom{N}{B}} \left [ \binom{N-k}{B-(i+1)} + \binom{N-k}{N-B-(i+1)} \right ] \label{eq:pmf_W_k_ver2} .
\end{align}
Here is the interpretation for this expression: Suppose we are just ordering the $N$ balls from 1 to $N$, there are a total of $\binom{N}{B}$ possible ways (the denominator). The event $\{W_k=i\}$ means there are $i$ balls in the first $k-1$ balls the same color as the $i$th ball. The term $\binom{k-1}{i}$ is the number of ways to choose the locations of these $i$ balls from $k-1$ spots. Next, the remaining $N-k$ spots have a total of $B-(i+1)$ blue balls if the $k$th ball is blue, and $N-B-(i+1)$ red balls of the $k$th ball is red, hence the result. \\

\noindent Using Equation (\ref{eq:pmf_W_k}) and the expected value of Hypergeometric random variables, we have \[ \mathbb{E}[W_k] = \frac{B}{N} \times \frac{(k-1)(B-1)}{(N-1)} + \frac{N-B}{N} \times \frac{(k-1)(N-B-1)}{(N-1)} .\] Then for $\alpha=1$, 
\begin{align*}
     \mathbb{E}[S] &= \sum_{k=3}^N \mathbb{E}\left [\frac{W_k}{k-1} \right ] \\
     &= \frac{(N-2)\big (B(B-1)+ (N-B)(N-B-1)\big )}{N(N-1)} \\
     &= (N-2) - \frac{2B(N-B)(N-2)}{N(N-1)}.
\end{align*}
For $\alpha > 1$, we can use our alternative formulation Equation (\ref{eq:pmf_W_k_ver2}). 
\begin{align*}
    \mathbb{E}\left [ \frac{\alpha W_k}{(k-1-W_k) + \alpha W_k} \right ] &= \sum_{i=I_{\min}}^{I_{\max}} \frac{\alpha i}{(k-1-i) + \alpha i} \frac{\binom{k-1}{i}}{\binom{N}{B}} \left [ \binom{N-k}{B-(i+1)} + \binom{N-k}{N-B-(i+1)} \right ] \\
    &= \frac{B}{N} \times \sum_{i=I_{\min}}^{I_{\max}} \frac{\alpha i}{(k-1-i) + \alpha i} \frac{\binom{k-1}{i} \binom{N-k}{B-(i+1)}}{\binom{N-1}{B-1}} + \\
    &\qquad  \frac{N-B}{N}\times \sum_{i=I_{\min}}^{I_{\max}} \frac{\alpha i}{(k-1-i) + \alpha i} \frac{\binom{k-1}{i} \binom{N-k}{N-B-(i+1)}}{\binom{N-1}{N-B-1}} .
\end{align*}
Here, $I_{\min}, I_{\max}$ just denote the support. With the convention $\binom{a}{b}=0$ if $a<b$, we can take $I_{\min}=0, I_{\max}= k-1$ for each $k$. This can be calculated numerically. 

\newpage 

\subsection{Algorithm to reconstruct the sequence of attachments, initial color order, and the order of tips added.} \label{subsec:tree_alg}

\begin{algorithm}[H]
\caption{Algorithm to reconstruct $C$ and the sequence of attachments from $\tilde{T}^\ell_{N,B}$}
\label{alg:get_attachments}
\KwData{}
Color labels $L=\{L_{t_1},..., L_{t_N}\}$: colors of the tips $t_1,..., t_N$ on the planar tree from left to right. \\
$i=1,...,N-1$: ranks of internal nodes from bottom to top. \\
    
\vspace{0.1in}

\KwResult{}
$A= [ \;\; ]$: an ordered list of the tips added. \\
$C = [\;\; ]$: an ordered list of the colors added. \\ 
$D= [\;\; ]$: an ordered list of the tips that have been attached to. \\
Integer $S$: the number of same attachments. \\
Array $\{w_k:k=3,...,N\}$. \\ 

\vspace{0.1in}

Initialize $S =0$ and $t_k=0$ for $k=1,...,N$. \\
\For{$i=1,...,N-1$}{
  Compute $l_i \in \{1,...,N\}$: the right-most leaf of the left subtree of internal node $i$. \\
  Compute $r_i \in \{1,...,N\}$: the right-most leaf of the right subtree of internal node $i$ \\
  Set $t_{l_i}=N+1-i$ and $C_{N+1-i} = L_{l_i}$. \\
  Set $A_i = t_{l_i}$: this is the tip being added. \\
  Set $D_i= t_{r_i}$: this is the tip being attached to. \\
  \If{$L_{l_i} = L_{r_i}$}{$S \gets S+1$}
}
Compute $w_k = \sum_{i=1}^{k-1} \mathds{1}(C_i=C_k)$ for each $k=3,...,N$. \\
\end{algorithm}

\subsection{Conditions for a list of tables to be valid}\label{subsec:table_conditions}
Let $\{E^1, E^2,...,E^T\}$ be the list of tables under consideration. To check it corresponds to a ranked planar partially labeled tree with $N$ tips, it must satisfy the following conditions. As per standard terminology, let $|E^t|$ denote the cardinality of $E^t$, which is the length of the table $t$. 
\begin{enumerate}
    \item The total number of lists $T$ satisfies $1\leq T \leq N$. 
    \item The total length of the lists $|E|=\sum_{t=1}^T |E^t|$ is either $|E|=N +T-1$ or $|E|=N+T-2$.
    \item Each element $\{1,...,N\}$ must appear at least once in $\{E^1, E^2,...,E^T\}$. 
    \item There are no repeat elements in $E^t$ for all $t=1,...,T$.
    \item If $E^t$ has length 1, then $E^t$ can be only $(1)$ or $(2)$. 
    \item For $E^t = (E^t_1,...,E^t_n)$ where $n\geq 2$: for all $i=1,...,n-1$, there exists $j$ such that $i<j$ and $E^t_i > E^t_j$. That is, there exists a number smaller than $E^t_i$ to the right of $E^t_i$. For example, if $|E^t| =2$, then $E^t_1 > E^t_2$. 
    \item The element 1 cannot appear to the left of any element in any list. 
    \item The element 2 can appear to the left of 1 only if $E^t=(1)$ does not exist. 
    \item Elements $3,...,N$ appear to the left of some elements in exactly one list. 
\end{enumerate}

\subsection{Details of tree topology simulation}\label{subsec:tree_topology_details}
We give the details of the two simulations conducted in Section~\ref{subsec:tree_topology}. To simulate the distribution of cherries and pitchforks, we
simulated 200 trees per $(N,B,\alpha)$ combination with $N \in \{4,5,6,7,8,9,10, 20, 50, 100, 200\}$, all values of $B \in [0,N/2]$, $\alpha \in \{ 1, 6, 11, 16,21, 50, 200\}$. In Figure~\ref{fig:cherry_pitchfork}(a), we see the number of cherries$/N$ is concentrated around $1/3$. In Figure~\ref{fig:cherry_pitchfork}(b), the number of pitchforks$/N$ is concentrated around $1/6$. For the second simulation, we generated 500 trees per $(N,B,\alpha)$ combination with $\alpha \in \{1,5, 25\}$ and $N=10, B=4$ or $N=100, B=10$. Then for each $N$, we apply MDS using distances on the F-matrices of these trees \cite{Kim2020}. In Figure~\ref{fig:mds_plot}, we see that there is no clustering by alpha for either $N=10$ or $N=100$. 

\begin{figure}[H]
    \centering
    \includegraphics[width=0.95\textwidth]{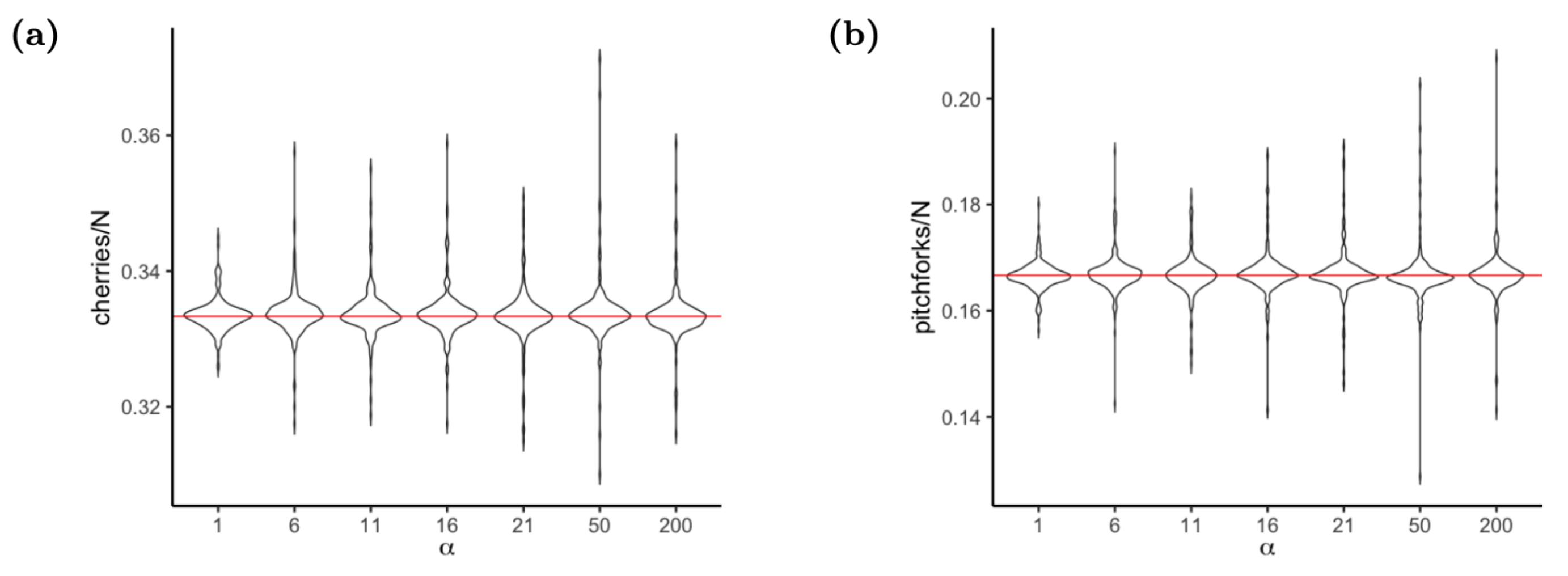}
    \caption{\textbf{Violin plots of the simulated number of cherries and pitchforks.} For each $\alpha$, the values are concentrated around the theoretical value. In (a), the red line is 1/3, which is the expected value of \# cherries$/N$. In (b), the red line is 1/6, which is the expected value of \# pitchforks$/N$.}
    \label{fig:cherry_pitchfork}
\end{figure}

\begin{figure}[H]
    \centering
    \includegraphics[width=0.95\textwidth]{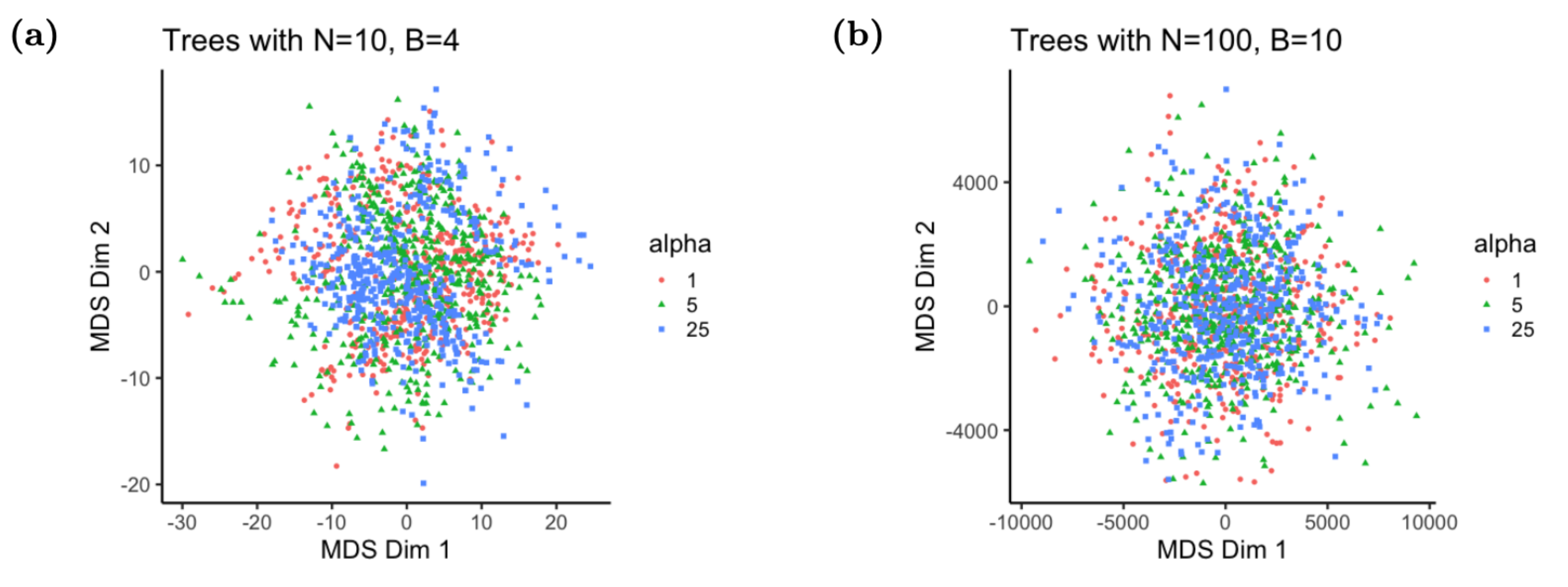}
    \caption{Plots of the first 2 MDS axes for (a) $N=10, B=4$ and (b) $N=100, B=10$}
    \label{fig:mds_plot}
\end{figure}

\newpage 

\subsection{Validity of p-values}\label{subsec:lemma_conv}
\begin{lemma}\label{lemma:indicator_conv}
If $X_n \xrightarrow{p} X$ and $c$ is a fixed constant, then $\mathds{1}(X_n \geq c) \xrightarrow{p} \mathds{1}(X\geq c)$. 
\end{lemma}
\begin{proof}
We only need to consider $\epsilon\in [0,1]$ because we are working with indicators. 
\begin{align*}
    \mathbb{P}( | \mathds{1}(X_n\geq c) - \mathds{1}(X\geq c) | \geq \epsilon ) &= \mathbb{P}\Big ( \mathds{1} (\{ X_n \geq c > X \} \cup \{ X \geq c > X_n\}) \geq \epsilon \Big ) \\
    &=  \mathbb{P}\Big ( \{ X_n \geq c > X \} \cup \{ X \geq c > X_n\}\Big )\\
    &\leq \mathbb{P}(  X_n -X \geq \epsilon_c) + \mathbb{P}( X-X_{n} \geq \epsilon_c)  \; \text{ for some $\epsilon_c>0$} \\
    &\to 0
\end{align*}
Therefore, we have convergence in probability.
\end{proof}

\subsection{Power Analyses on Specific Trees}
The following figures show the power approximations calculated on five fixed ranked tree shapes detailed in Section~\ref{subsec:power_sim}. 

\begin{figure}[H]
  \begin{minipage}{0.45\linewidth}
    \centering
    \includegraphics[width=\linewidth]{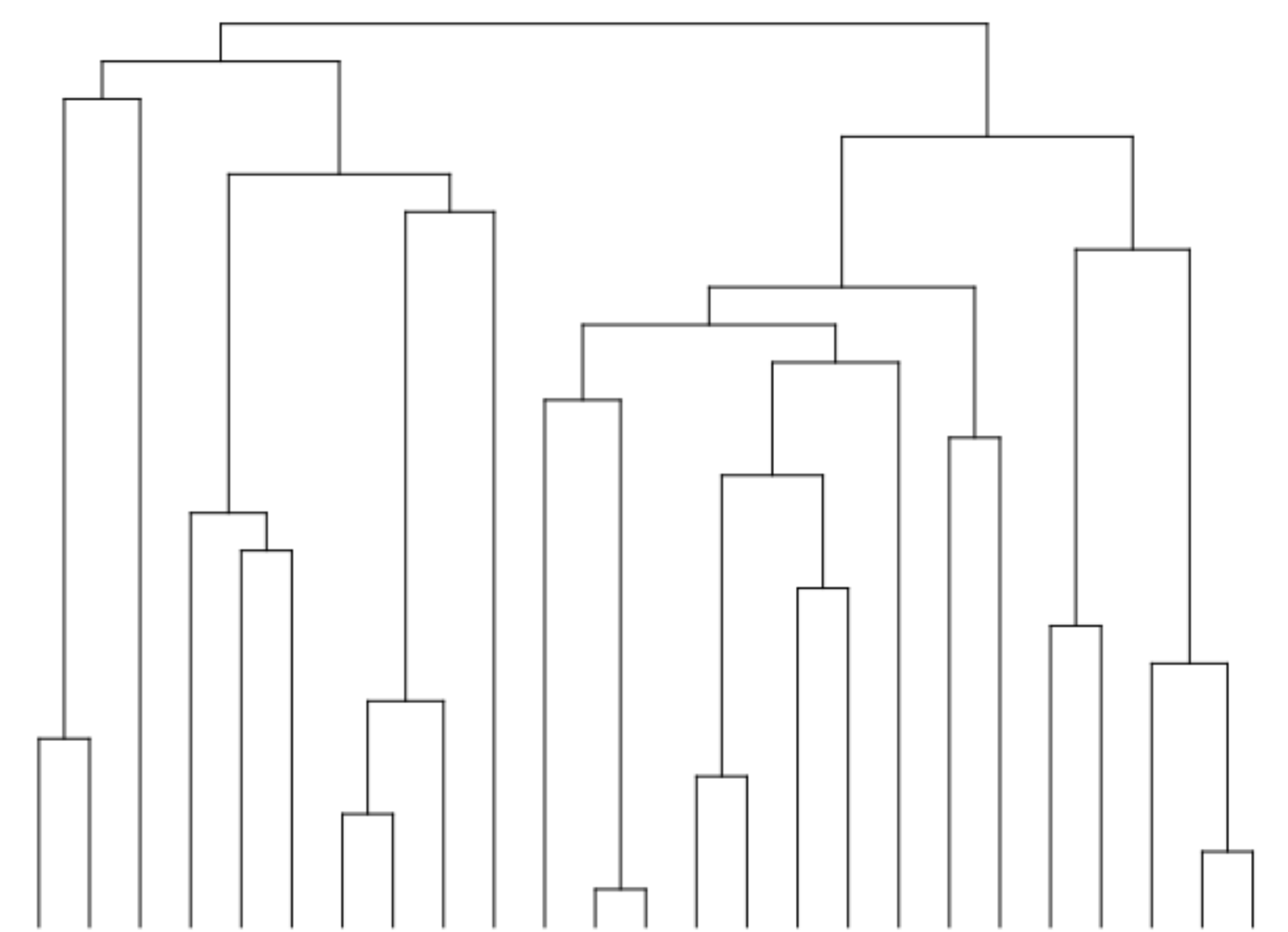}
  \end{minipage}%
  \begin{minipage}{0.6\linewidth}
    \centering
    \begin{tabular}{|r| r|rrrrr|}
    \multicolumn{7}{c}{Probability of rejecting $\alpha=1$} \\
  \hline
 $B/N$ & $\alpha_0$ & $S$ & $\hat{\mu}$ & AI & PS & treeSeg \\ 
  \hline
   \multirow{5}{*}{0.1} & 2 & 0.406 & 0.19 & 0.046 & 0.152 & 0  \\ 
   & 5 & 0.666 & 0.543 & 0.056 & 0.212 & 0  \\ 
    &10 & 0.765 & 0.842 & 0.054 & 0.104 & 0 \\ 
    &20 &  0.873 & 0.988 & 0.052 & 0.184 & 0\\ 
   \hline%
\multirow{5}{*}{0.25} & 2 & 0.165 & 0.156 & 0.05 & 0.182 & 0.008 \\ 
   & 5 & 0.338 & 0.419 & 0.068 & 0.202 & 0.006 \\ 
    &10 & 0.778 & 0.922 & 0.315 & 0.561 & 0.032  \\ 
    &20 & 0.732 & 0.83 & 0.126 & 0.421 & 0 \\ 
   \hline%
   \multirow{5}{*}{0.4} & 2 & 0.247 & 0.196 & 0.068 & 0.12 & 0.002 \\ 
   & 5 & 0.408 & 0.351 & 0.08 & 0.11 & 0 \\ 
    &10 &  0.686 & 0.677 & 0.14 & 0.16 & 0.008 \\ 
    &20 & 0.736 & 0.78 & 0.076 & 0.16 & 0.004  \\ 
   \hline%
   \multirow{5}{*}{0.5} & 2 &  0.221 & 0.23 & 0.098 & 0.086 & 0 \\ 
   & 5 &  0.337 & 0.377 & 0.062 & 0.064 & 0.004 \\ 
    &10 & 0.657 & 0.796 & 0.132 & 0.09 & 0.022\\ 
    &20 & 0.522 & 0.607 & 0.08 & 0.062 & 0.002 \\ 
    \hline
\end{tabular}

\end{minipage}
\caption{Power Calculation for testing $H_0:\alpha=1$ vs $H_1:\alpha=\alpha_0$ for the random tree shape with $N=25$ and various values of $B$. }
\label{fig:tree_25_random_power}
\end{figure}

\begin{figure}[H]
  \begin{minipage}{0.45\linewidth}
    \centering
    \includegraphics[width=\linewidth]{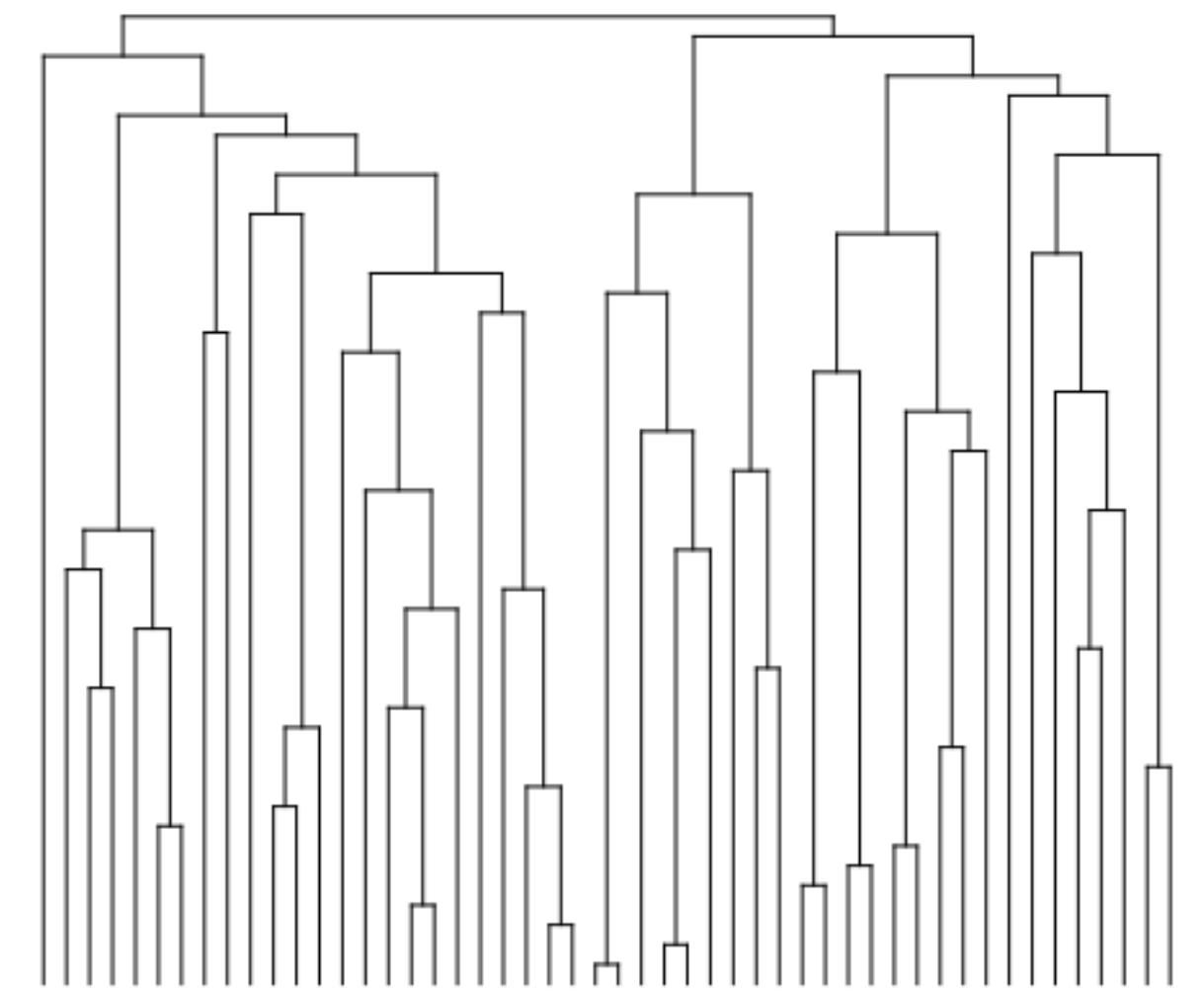}
  \end{minipage}%
  \begin{minipage}{0.6\linewidth}
    \centering
\begin{tabular}{|r| r|rrrrr|}
\multicolumn{7}{c}{Probability of rejecting $\alpha=1$} \\
  \hline
 $B/N$ & $\alpha_0$ & $S$ & $\hat{\mu}$ & AI & PS & treeSeg \\ 
  \hline
   \multirow{5}{*}{0.1} & 2 & 0.201 & 0.046 & 0.03 & 0.124 & 0 \\ 
   & 5 & 0.516 & 0.379 & 0.084 & 0.17 & 0 \\ 
    &10 & 0.736 & 0.858 & 0.11 & 0.291 & 0 \\ 
    &20 & 0.818 & 0.952 & 0.082 & 0.186 & 0\\ 
   \hline%
   \multirow{5}{*}{0.25} & 2 & 0.2 & 0.226 & 0.04 & 0.164 & 0.002  \\ 
    &5 & 0.7 & 0.91 & 0.132 & 0.335 & 0 \\ 
    &10 & 0.829 & 0.994 & 0.156 & 0.279 & 0.004 \\ 
    &20 & 0.929 & 1 & 0.152 & 0.307 & 0.018 \\ 
    \hline 
    \multirow{5}{*}{0.4} & 2 & 0.281 & 0.285 & 0.076 & 0.098 & 0 \\ 
    &5 &  0.635 & 0.735 & 0.076 & 0.098 & 0 \\ 
    &10 & 0.957 & 0.998 & 0.267 & 0.266 & 0.006 \\ 
    &20 & 0.917 & 0.982 & 0.094 & 0.142 & 0.006 \\ 
    \hline 
    \multirow{5}{*}{0.5} & 2 & 0.403 & 0.379 & 0.124 &0.162 & 0 \\ 
    &5 & 0.803 & 0.874 & 0.098 & 0.226 & 0.004 \\ 
    &10 & 0.749 & 0.782 & 0.044 & 0.108 & 0 \\ 
    & 20 & 0.945 & 0.972 & 0.096 & 0.225 & 0.014  \\ 
    \hline
\end{tabular}

\end{minipage}
\caption{Power Calculation for testing $H_0:\alpha=1$ vs $H_1:\alpha=\alpha_0$ for the random tree shape with $N=50$ and various values of $B$. }
\label{fig:tree_50_random_power}
\end{figure}

\begin{figure}[H]
  \begin{minipage}{0.45\linewidth}
    \centering
    \includegraphics[width=\linewidth]{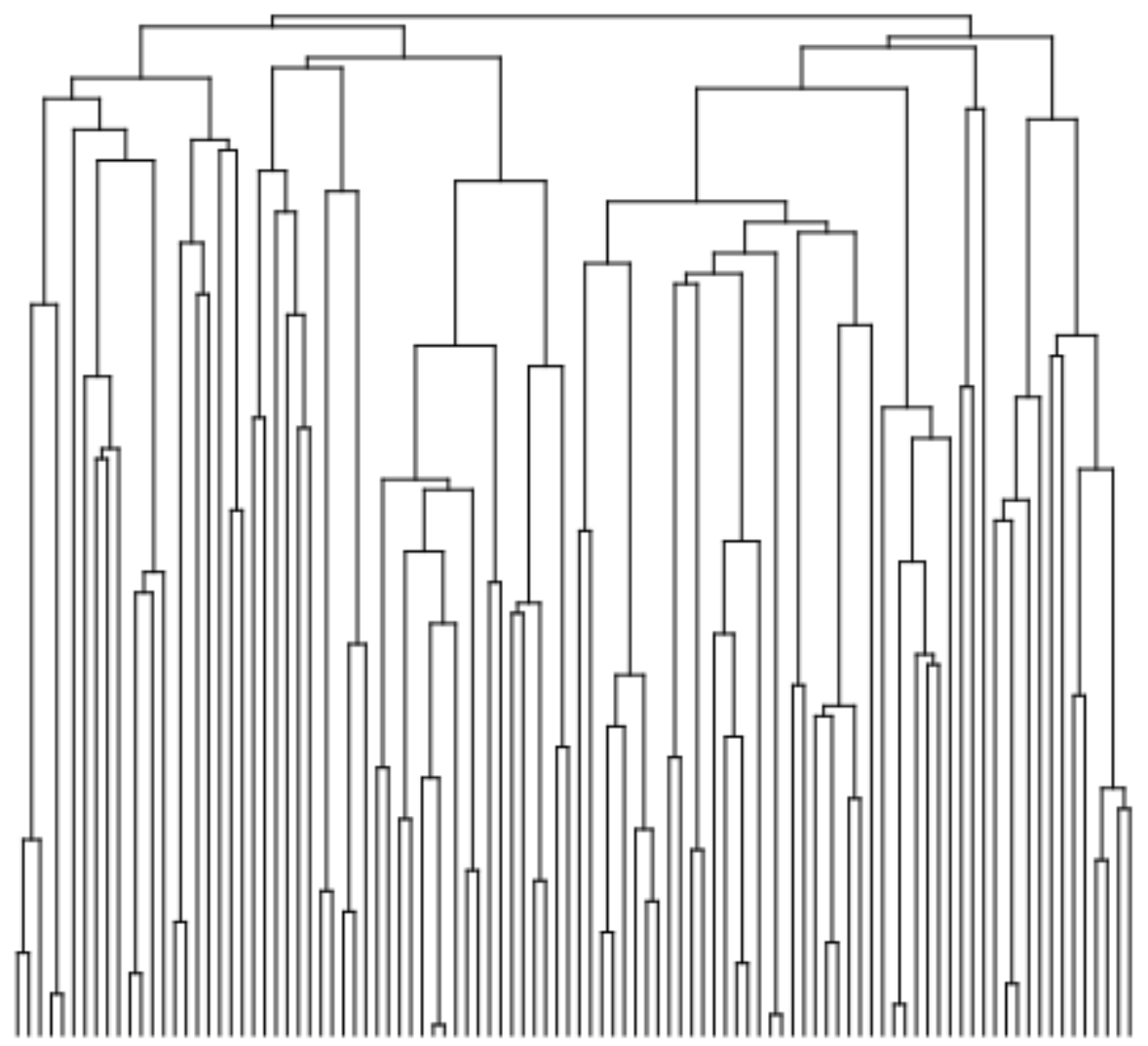}
  \end{minipage}%
  \begin{minipage}{0.6\linewidth}
    \centering
\begin{tabular}{|r| r|rrrrr|}
\multicolumn{7}{c}{Probability of rejecting $\alpha=1$} \\
  \hline
 $B/N$ & $\alpha_0$ & $S$ & $\hat{\mu}$ & AI & PS & treeSeg \\ 
  \hline
   \multirow{5}{*}{0.1} & 2 & 0.315 & 0.391 & 0.074 & 0.451 & 0 \\ 
   & 5 & 0.755 & 0.978 & 0.066 & 0.477 & 0.002 \\ 
    &10 & 0.862 & 1 & 0.066 & 0.419 & 0.002 \\ 
    &20 & 0.955 & 1 & 0.144 & 0.405 & 0 \\ 
   \hline%
   \multirow{5}{*}{0.25} & 2 & 0.518 & 0.661 & 0.058 & 0.214 & 0.002 \\ 
   & 5 & 0.91 & 1 & 0.07 & 0.188 & 0  \\ 
    &10 & 0.953 & 1 & 0.092 & 0.204 & 0.004 \\ 
    &20 & 0.964 & 1 & 0.08 & 0.242 & 0 \\ 
   \hline%
   \multirow{5}{*}{0.4} & 2 & 0.531 & 0.577 & 0.098 & 0.088 & 0.008 \\ 
   & 5 & 0.983 & 1 & 0.202 & 0.228 & 0.004 \\ 
    &10 & 0.931 & 0.998 & 0.06 & 0.09 & 0.006 \\ 
    &20 & 0.937 & 1 & 0.098 & 0.11 & 0.001 \\ 
   \hline%
   \multirow{5}{*}{0.5} & 2 & 0.53 & 0.635 & 0.082 & 0.134 & 0 \\ 
   & 5 & 0.907 & 0.99 & 0.096 & 0.202 & 0  \\ 
    &10 & 0.929 & 0.992 & 0.094 & 0.174 & 0  \\ 
    &20 & 0.93 & 0.998 & 0.072 & 0.14 & 0.002 \\ 
   \hline%
\end{tabular}
\end{minipage}
\caption{Power Calculation for testing $H_0:\alpha=1$ vs $H_1:\alpha=\alpha_0$ for the random tree shape with $N=100$ and various values of $B$. }
\label{fig:tree_100_random_power}
\end{figure}

\begin{figure}[H]
  \begin{minipage}{0.45\linewidth}
    \centering
    \includegraphics[width=\linewidth]{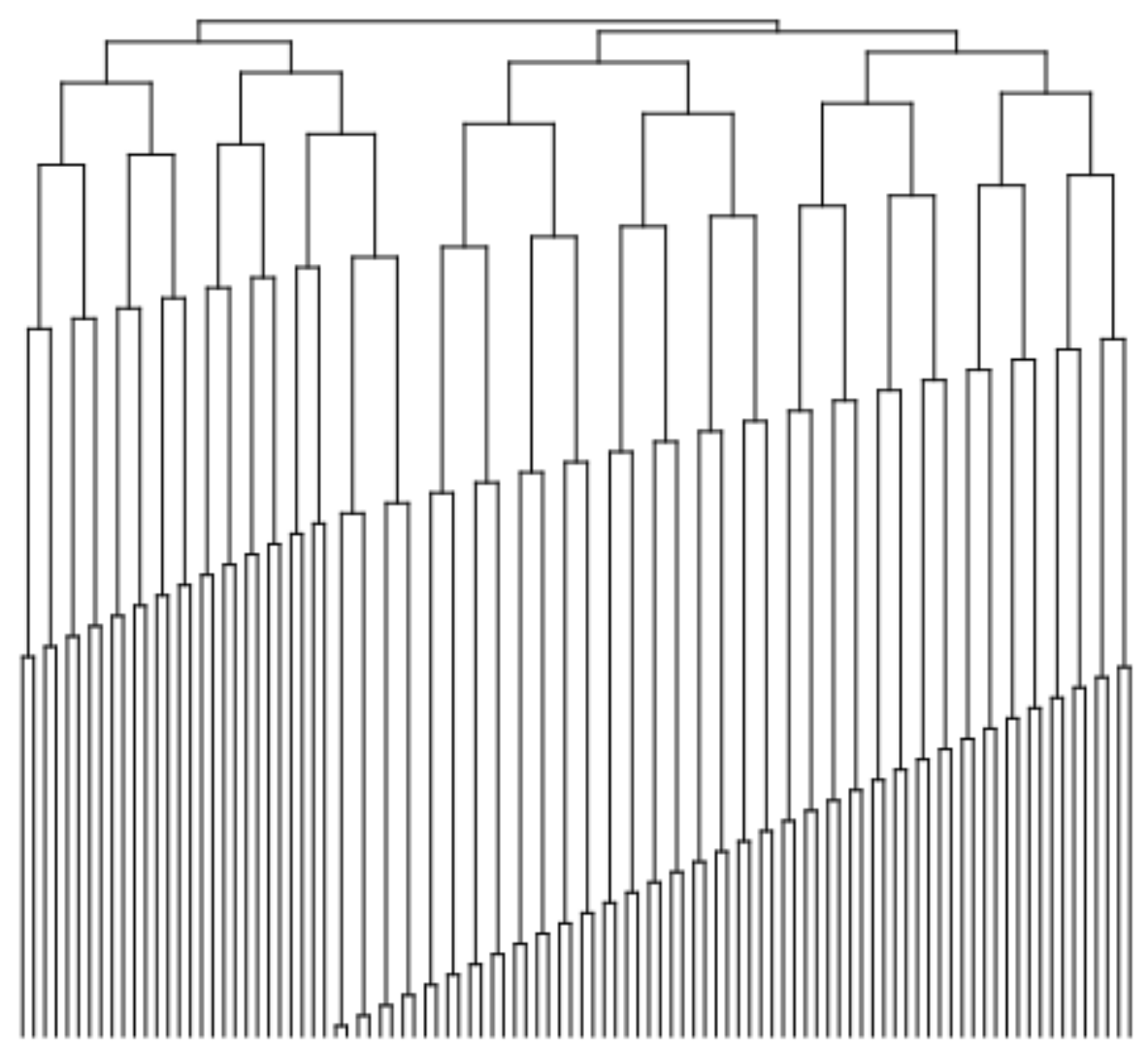}
  \end{minipage}%
  \begin{minipage}{0.6\linewidth}
    \centering
    \begin{tabular}{|r| r|rrrrr|}
    \multicolumn{7}{c}{Probability of rejecting $\alpha=1$} \\
  \hline
  $B/N$ & $\alpha_0$ & $S$ & $\hat{\mu}$ & AI & PS & treeSeg\\ 
  \hline
   \multirow{5}{*}{0.1} & 2 & 0.324 & 0.615 & 0.176 & 0.18 & 0  \\ 
   & 5 &  0.731 & 1 & 0.106 & 0.106 & 0 \\ 
    &10 &  0.884 & 1 & 0.108 & 0.108 & 0 \\ 
    &20 & 0.897 & 1 & 0.126 & 0.126 & 0.002 \\ 
   \hline%
\multirow{5}{*}{0.25} & 2 & 0.322 & 0.491 & 0.082 & 0.056 & 0.002  \\ 
   & 5 & 0.804 & 0.982 & 0.1 & 0.068 & 0.002  \\ 
    &10 & 0.97 & 1 & 0.222 & 0.17 & 0  \\ 
    &20 & 0.918 & 1 & 0.1 & 0.084 & 0.008 \\ 
   \hline%
   \multirow{5}{*}{0.4} & 2 & 0.321 & 0.445 & 0.058 & 0.144 & 0.006  \\ 
   & 5 &  0.941 & 1 & 0.23 & 0.321 & 0  \\ 
    &10 & 0.936 & 1 & 0.112 & 0.255 & 0.018 \\ 
    &20 & 0.951 & 1 & 0.13 & 0.291 & 0.008 \\ 
   \hline%
   \multirow{5}{*}{0.5} & 2 & 0.398 & 0.531 & 0.036 & 0.056 & 0.006 \\ 
   & 5 & 0.878 & 0.99 & 0.064 & 0.092 & 0.014  \\ 
    &10 & 0.959 & 1 & 0.116 & 0.14 & 0 \\ 
    &20 &  0.893 & 0.994 & 0.042 & 0.064 & 0.004 \\ 
   \hline%
\end{tabular}
\end{minipage}
\caption{Power Calculation for testing $H_0:\alpha=1$ vs $H_1:\alpha=\alpha_0$ for the most-balanced tree shape with $N=100$ and various values of $B$. }
\label{fig:tree_100_bal_power}
\end{figure}

\begin{figure}[H]
  \begin{minipage}{0.45\linewidth}
    \centering
    \includegraphics[width=\linewidth]{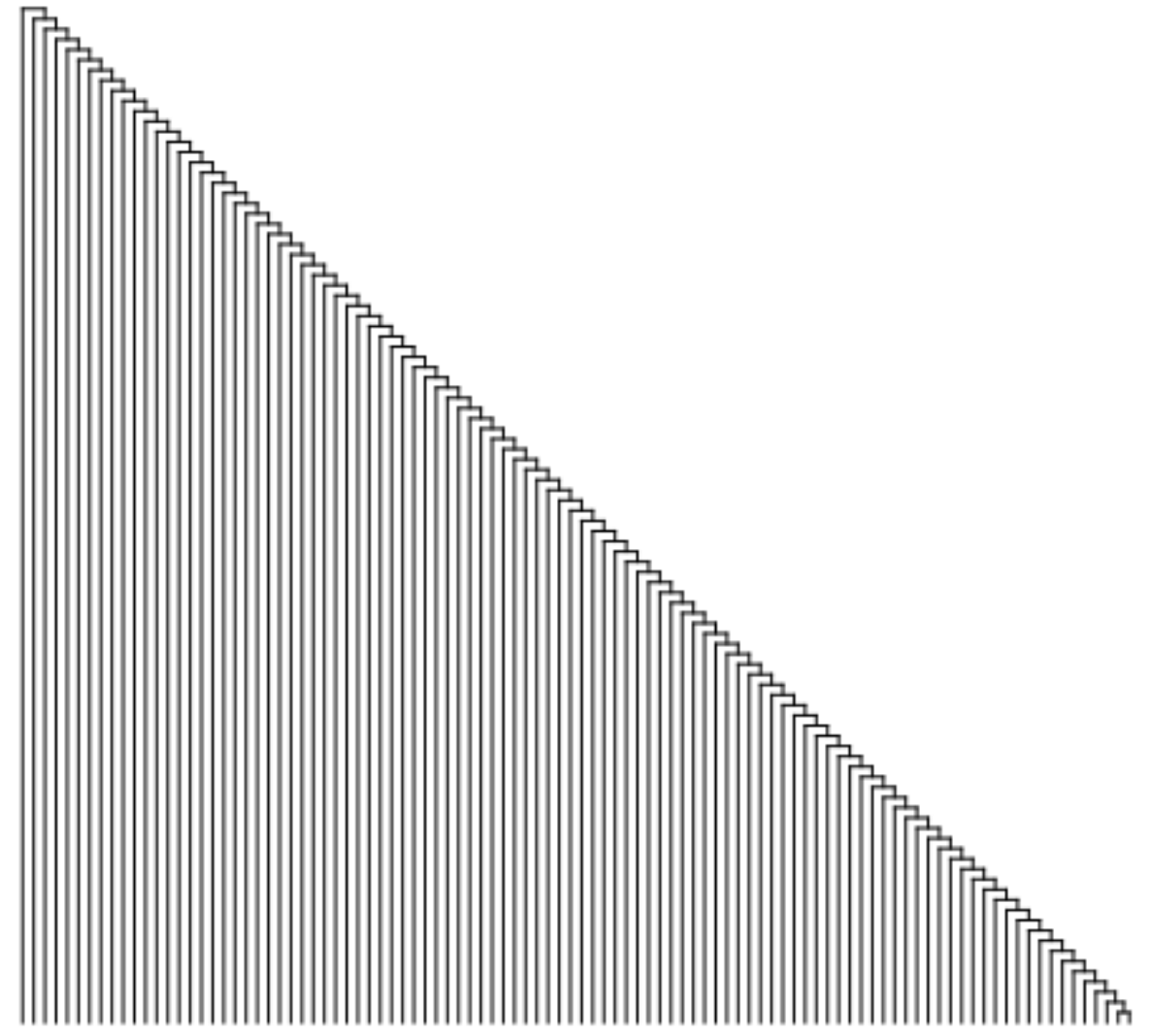}
  \end{minipage}%
  \begin{minipage}{0.6\linewidth}
    \centering    
    \begin{tabular}{|r| r|rrrrr|}
    \multicolumn{7}{c}{Probability of rejecting $\alpha=1$} \\
  \hline
  
 $B/N$ & $\alpha_0$ & $S$ & $\hat{\mu}$ & AI & PS & treeSeg \\ 
  \hline
   \multirow{5}{*}{0.1} & 2 & 0.032 & 0.004 & 0.048 & 0.08 & 0 \\ 
   & 5 & 0.062 & 0.09 & 0.036 & 0.08 & 0 \\ 
    &10 & 0.141 & 0.202 & 0.056 & 0.098 & 0 \\ 
    &20 & 0.277 & 0.331 & 0.088 & 0.066 & 0.002 \\ 
   \hline%
    \multirow{5}{*}{0.25} & 2 & 0.079 & 0.016 & 0.072 & 0.052 & 0.002 \\ 
   & 5 &  0.308 & 0.156 & 0.084 & 0.038 & 0  \\ 
    &10 & 0.495 & 0.327 & 0.08 & 0.038 & 0\\ 
    &20 & 0.584 & 0.421 & 0.062 & 0.058 & 0.006  \\ 
   \hline%
   \multirow{5}{*}{0.4} & 2 & 0.502 & 0.439 & 0.1 & 0.084 & 0.004  \\ 
   & 5 &  0.891 & 0.926 & 0.108 & 0.086 & 0 \\ 
    &10 & 0.944 & 0.978 & 0.094 & 0.068 & 0.008 \\ 
    &20 & 0.951 & 0.992 & 0.078 & 0.078 & 0.004 \\ 
   \hline%
   \multirow{5}{*}{0.5} & 2 & 0.958 & 0.984 & 0.098 & 0.064 & 0.002 \\ 
   & 5 &  0.995 & 1 & 0.076 & 0.062 & 0.012 \\ 
    &10 & 0.995 & 1 & 0.068 & 0.044 & 0 \\ 
    &20 & 0.996 & 1 & 0.08 & 0.044 & 0.002 \\ 
   \hline%
\end{tabular}
\end{minipage}
\caption{Power Calculation for testing $H_0:\alpha=1$ vs $H_1:\alpha=\alpha_0$ for the most-unbalanced tree shape with $N=100$ and various values of $B$.}
\label{fig:tree_100_unb_power}
\end{figure}

\newpage 

\subsection{DNA data simulation to compare BaTS and Posterior p-values}\label{subsec:dna} 
We present the details of the simulation described in Section~\ref{subsec:bayesian_sim}. We simulated two partially labeled ranked tree shapes with $N=50$ and $B=20$ from the CRPTree model with $\alpha=1$ and $\alpha=10$ respectively. We used the R-package phylodyn to simulate the branch lengths of the phylogenies according to the coalescent with exponentially growing effective population size \citep{karcher2017phylodyn}. We then used seqgen to simulate the $50$ molecular sequences of 100 nucleotides at the tips of each phylogeny according to the Jukes Cantor mutation model \citep{JukesCantor1969,rambaut1997seq}.\\

To estimate the two posterior phylogenetic distributions we used BEAST assuming the Jukes-Cantor mutation model with fixed mutation rate, a coalescent prior on the phylogenies, and a Gaussian Markov random field prior on $N_{e}(t)$ \citep{minin_smooth_2008}. We generated 100 billion iterations and thinned every 100 thousand iterations to obtain a posterior sample of 1000 ranked and partially labeled trees. \\

To compute the 1000 $p$-values with our method for each analysis, we used 500 label and planar permutations per tree. For BaTS analyses, we generated 500 samples by permuting the labels and compared the distribution of the posterior median statistics with the observed median statistic. The posterior distribution of p-values $p_T$ is shown in Figure~\ref{fig:bayesian_simulation}. For $\alpha=1$, the median value of $\hat{\mu}=20.31$ with a 95\% credible interval of $[19, 21.838 ]$. Using BaTS the 95\% credible interval for the posterior median is $[24.042, 24.582 ]$. For $\alpha=10$, the median value of $\hat{\mu}=39.36$ with a 95\% credible interval of $[39, 39.92 ]$. For BaTS, the 95\% credible interval for the posterior median is $[23.17, 25.54 ]$. It is clear that our method correctly rejects the case where $\alpha=10$, and fails to reject the case of $\alpha=1$, while BaTS would reject in both cases. 

\end{document}